\documentclass[11pt,english]{article}
\usepackage{amsmath,amssymb,amsthm}
\usepackage{multicol}
\usepackage[table]{xcolor}
\usepackage[margin=1in]{geometry}
\usepackage{graphicx,color}
\usepackage{enumitem}
\usepackage{fullpage}
\usepackage[noblocks]{authblk}
\usepackage{tcolorbox}
\usepackage{babel}
\usepackage{wrapfig}
\usepackage{MnSymbol}
\usepackage{mdframed}
\usepackage{mathtools}
\usepackage{floatrow}
\usepackage{multirow}
\usepackage{booktabs}
\usepackage{palatino}

\usepackage[ruled,linesnumbered,vlined]{algorithm2e}
\usepackage{tablefootnote}
\usepackage{thm-restate}
\usepackage{bbding}
\usepackage{pifont}
\usepackage{nicefrac}

\newcounter{magicrownumbers}

\graphicspath{{./figures/}}

\definecolor{Darkblue}{rgb}{0,0,0.4}
\definecolor{Brown}{cmyk}{0,0.61,1.,0.60}
\definecolor{Purple}{cmyk}{0.45,0.86,0,0}
\definecolor{Darkgreen}{rgb}{0.133,0.543,0.133}

\usepackage[colorlinks,linkcolor=Darkblue,filecolor=blue,citecolor=blue,urlcolor=Darkblue,pagebackref]{hyperref}
\usepackage[nameinlink]{cleveref}

\usepackage[colorinlistoftodos,prependcaption,textsize=tiny]{todonotes}

\newcommand{\namedref}[2]{\hyperref[#2]{#1~\ref*{#2}}}
\newcommand{\propref}[1]{\hyperref[#1]{property~(\ref*{#1})}}

\newtheorem*{theorem*}{Theorem}
\newtheorem{theorem}{Theorem}
\newtheorem{lemma}{Lemma}
\newtheorem{definition}{Definition}

\newtheorem*{question*}{Question}

\newtheorem*{conjecture*}{Conjecture}
\newtheorem{question}{Question}

\newcommand{\old}[1]{{}}

\hypersetup{
    colorlinks=true,
    linkcolor=blue,
    filecolor=violet,  
    citecolor=violet,
    urlcolor=cyan,
    pdftitle={Overleaf Example},
    pdfpagemode=FullScreen,
}

\title{Fully Dynamic Maximum Independent Sets of Disks\\ in Polylogarithmic Update Time
}

\date{}

\author{Sujoy Bhore\thanks{Department of Computer Science \& Engineering, Indian Institute of Technology Bombay, Mumbai, India.\\ Email: \texttt{sujoy@cse.iitb.ac.in}}
\quad
Martin N\"ollenburg\thanks{TU Wien, Vienna, Austria. Email: \texttt{noellenburg@ac.tuwien.ac.at}.}
 \quad
 Csaba D. T\'oth\thanks{Department of Mathematics, California State University Northridge, Los Angeles, CA; and Department of Computer Science, Tufts University, Medford, MA, USA. Email: \texttt{csaba.toth@csun.edu}}
 \quad
 Jules Wulms\thanks{Department of Mathematics and Computer Science, TU Eindhoven, Eindhoven, the Netherlands. Email: \texttt{j.j.h.m.wulms@tue.nl}} 
}

\begin{document}
\maketitle

\begin{abstract}
A fundamental question is whether one can 
maintain a maximum independent set in polylogarithmic update time for a dynamic collection of geometric objects in Euclidean space. 
Already, for a set of intervals, it is known that no dynamic algorithm can maintain an exact maximum independent set in sublinear update time. Therefore, the typical objective is to explore the trade-off between update time and solution size. Substantial efforts have been made in recent years to understand this question for various families of geometric objects, such as intervals, hypercubes, hyperrectangles, and fat objects. 

We present the first fully dynamic approximation algorithm for disks of arbitrary radii in the plane that maintains a constant-factor approximate maximum independent set in polylogarithmic expected amortized update time.  
Moreover, for a fully dynamic set of $n$ disks of unit radius in the plane, we show that a $12$-approximate maximum independent set can be maintained with worst-case update time $O(\log n)$, and optimal output-sensitive reporting. This result generalizes to fat objects of comparable sizes in any fixed dimension $d$, where the approximation ratio depends on the dimension and the fatness parameter. 
Further, we note that, even for a dynamic set of disks of unit radius in the plane, it is impossible to maintain $O(1+\varepsilon)$-approximate maximum independent set in truly sublinear update time, under standard complexity assumptions.

Our results build on two recent technical tools: (i) The MIX algorithm by Cardinal et al.\ (ESA~2021) that can smoothly transition from one independent set to another; hence it suffices to maintain a family of independent sets where the largest one is a constant-factor approximation of a maximum independent set. (ii) A dynamic nearest/farthest neighbor data structure for disks by Kaplan et al.\ (DCG~2020) and Liu (SICOMP~2022), which generalizes the dynamic convex hull data structure by Chan (JACM~2010), and allows us to quickly find a ``replacement'' disk (if any) when a disk in one of our independent sets is deleted.
\end{abstract}

\section{Introduction}
The maximum independent set (\textsc{MIS}) problem is one of the most fundamental problems in theoretical computer science, and it is one of Karp's 21 classical \textsf{NP}-complete problems~\cite{Karp72}.
In the \textsc{MIS} problem, we are given a graph $G=(V, E)$, and the objective is to choose a subset of the vertices $S\subseteq V$ of maximum cardinality such that no two vertices in $S$ are adjacent. The intractability of \textsc{MIS} carries even under strong algorithmic paradigms. For instance, it is known to be hard to approximate: no polynomial-time algorithm can achieve an approximation factor $n^{1-\varepsilon}$ (for $|V|=n$ and a constant $\varepsilon >0$) unless \textsf{P=ZPP}~\cite{Zuckerman07}. In fact, even if the maximum degree of the input graph is bounded by $3$, no polynomial-time approximation scheme (\textsf{PTAS}) is possible~\cite{berman1999approximation}.

\medskip \noindent \textbf{Geometric Independent Set.} 
In geometric settings, the input to the MIS problem is a collection $\mathcal{L}=\{\ell_1,\ldots,\ell_n\}$ of geometric objects, e.g., intervals, disks, squares, rectangles, etc., and we wish to compute a maximum independent set in their intersection graph $G$. That is, each vertex in $G$ corresponds to an object in $\mathcal L$, and two vertices form an edge if and only if the two corresponding objects intersect. The objective is to choose a maximum cardinality subset $\mathcal{L'}\subseteq \mathcal{L}$ of independent (i.e., pairwise disjoint) objects. 

A large body of work has been devoted to the \textsc{MIS} problem in geometric settings, due to its wide range of applications in scheduling~\cite{bar2006scheduling}, VLSI design~\cite{HochbaumM85}, map labeling~\cite{agarwal1998label}, data mining~\cite{khanna1998approximating, BermanDMR01}, and many others. Stronger theoretical results are known for the \textsc{MIS} problem in the geometric setting, in comparison to general graphs. For instance, even for unit disks in the plane,  the problem remains \textsf{NP}-hard~\cite{ClarkCJ90} and \textsf{W[1]}-hard~\cite{Marx05}, but it
admits a \textsf{PTAS}~\cite{HochbaumM85}. Later, \textsf{PTAS}s were also developed for arbitrary disks, squares, and more generally hypercubes and fat objects in constant dimensions~\cite{HuntMRRRS98,Chan03,AlberF04,erlebach2005polynomial}. 

In their seminal work, Chan and Har-Peled~\cite{ChanH12} showed that for an arrangement of pseudo-disks,\footnote{A set of objects is an arrangement of pseudo-disks if the boundaries of every pair of them intersect at most twice.} a local-search-based approach yields a \textsf{PTAS}. However, for non-fat objects, the scenario is quite different. For instance, it had been a long-standing open problem to find a constant-factor approximation algorithm for the \textsc{MIS} problem on axis-aligned rectangles. In a recent breakthrough, Mitchell~\cite{mitchell2022approximating} answered this question in the affirmative. Through a refined analysis of the recursive partitioning scheme, a dynamic programming algorithm finds a constant-factor approximation. Subsequently, G\'alvez et al.~\cite{galvez20223} improved the approximation ratio to $3$.  

\medskip \noindent \textbf{Dynamic Geometric Independent Set.}
In dynamic settings, objects are inserted into or deleted from the collection $\mathcal L$ over time. The typical objective is to achieve (almost) the same approximation ratio as in the offline (static) case while keeping the update time (i.e., the time to update the solution after insertion/deletion) as small as possible. We call it the \emph{Dynamic Geometric Maximum Independent Set} problem (for short, \textsc{DGMIS}). 

Henzinger et al.~\cite{Henzinger0W20} studied \textsc{DGMIS} for various geometric objects, such as intervals, hypercubes, and hyperrectangles. Many of their results extend to the weighted version of \textsc{DGMIS}, as well. Based on a lower bound of Marx~\cite{marx2007optimality} for the offline problem, they showed that any dynamic $(1+\varepsilon)$-approximation for squares in the plane requires $\Omega(n^{1/\varepsilon})$ update time for any $\varepsilon>0$, ruling out the possibility of sub-polynomial time dynamic approximation schemes. On the positive side, they obtained dynamic algorithms with update time polylogarithmic in both $n$ and $N$, 
where the corners of the objects are in a $[0, N]^d$ integer grid, for any constant dimension $d$ (therefore their aspect ratio is also bounded by $N$). Gavruskin et al.~\cite{gavruskin2015dynamic} studied \textsc{DGMIS} for intervals in $\mathbb{R}$
under the assumption that no interval is contained in another interval and obtained an optimal solution with $O(\log n)$ amortized update time. Bhore et al.~\cite{bhore2020dynamic} presented the first fully dynamic algorithms with polylogarithmic update time for \textsc{DGMIS}, where the input objects are intervals and axis-aligned squares.
For intervals, they presented a fully dynamic $(1+\varepsilon)$-approximation algorithm with logarithmic update time. Later, Compton et al.~\cite{compton2020new} achieved a faster update time for intervals, by using a new partitioning scheme. Recently, Bhore et al.~\cite{BhoreKO22} studied the MIS problem for intervals in the streaming settings, and obtained lower bounds. 

For axis-aligned squares in $\mathbb{R}^2$,, Bhore et al.~\cite{bhore2020dynamic} presented a randomized algorithm with an expected approximation ratio of roughly $2^{12}$ (generalizing to roughly $2^{2d+5}$ for $d$-dimensional hypercubes) with amortized update time $O(\log^5 n)$ (generalizing to $O(\log^{2d+1} n)$ for hypercubes). Moreover, Bhore et al.~\cite{bhore2022algorithmic} studied the \textsc{DGMIS} problem in the context of dynamic map labeling and presented dynamic algorithms for several subfamilies of rectangles that also perform well in practice. Cardinal et al.~\cite{CardinalIK21} designed dynamic algorithms for fat objects in fixed dimension $d$ with sublinear worst-case update time. Specifically, they achieved $\tilde{O}(n^{3/4})$ update time\footnote{The $\tilde{O}(\cdot)$ notation ignores logarithmic factors.} for disks in the plane, and $\tilde{O}(n^{1-\frac{1}{d+2}})$ for Euclidean balls in $\mathbb{R}^d$.

However, despite the remarkable progress on the \textsc{DGMIS} problem in recent years, the following question remained unanswered. 

\vspace{.2cm}


\begin{tcolorbox}
{\begin{question}\label{prob1}
Does there exist an algorithm that, for a given dynamic set of disks in the plane, maintains a constant-factor approximate maximum independent set in polylogarithmic update time?
\end{question}}
\end{tcolorbox}

%
%
%
%

\subsection*{Our Contributions}\label{sec:cont} 
In this paper, we answer Question~\ref{prob1} in the affirmative (Theorems~\ref{thm:UnitDisks}--\ref{thm:ArbitraryDisks}); see Table~\ref{table}. 
As a first step, we address the case of unit disks in the plane.

\begin{table}[ht]
\begin{tabular}{llll}
\toprule

\multicolumn{1}{c}{Objects}  & \multicolumn{1}{c}{Approximation Ratio} & \multicolumn{1}{c}{Update time} & \multicolumn{1}{c}{Reference} \\ 

\midrule

\rowcolor[gray]{.9} Intervals  & $1+\varepsilon$ & $O(\varepsilon^{-1} \log n)$  & \cite{compton2020new} \\ 

Squares  & $O(1)$ & $O(\log^5 n)$ amortized & \cite{bhore2020dynamic} \\ 

\rowcolor[gray]{.9}  Arbitrary radii disks  & $O(1)$ & $(\log n)^{O(1)}$ expec. amortized & Theorem~\ref{thm:ArbitraryDisks} \\ 

 \multirow{2}{*}{Unit disks}  & $O(1)$ & $O(\log n)$ worst-case& Theorem~\ref{thm:UnitDisks} \\ 

 & $1+\varepsilon$ & $n^{(1/\varepsilon)^{\Omega(1)}}$ & Theorem~\ref{lb:disk} \\ 

\rowcolor[gray]{.9} $f$-fat objects in $\mathbb{R}^d$ & $O_{f,d}(1)$ & $O_{f.d}(\log n)$ worst-case& Theorem~\ref{thm:Fat} \\ 

$d$-dimensional hypercubes & $(1+\varepsilon)\cdot 2^d$ & $O_{d,\varepsilon}(\log^{2d+1} n \cdot\log^{2d+1} U)$ & \cite{Henzinger0W20} \\

\bottomrule
\end{tabular}
\caption{\small Summary of results on dynamic independent sets for geometric objects.\label{table}}
\end{table}

\begin{restatable}{theorem}{theoremUnitDisk}
\label{thm:UnitDisks}
For a fully dynamic set of unit disks in the plane, a 12-approximate MIS can be maintained with worst-case update time $O(\log n)$, and optimal output-sensitive reporting.  
\end{restatable}

We prove Theorem~\ref{thm:UnitDisks} in Section~\ref{sec:unit}. Similarly to classical approximation algorithms for the static version~\cite{HochbaumM85}, we lay out four shifted grids such that any unit disk lies in a grid cell for at least one of the grids. For each grid, we maintain an independent set that contains at most one disk from each grid cell, thus we obtain four independent sets $S_1,\ldots , S_4$ at all times. Moreover, the largest of $S_1,\ldots , S_4$ is a constant-factor approximation of the MIS (Lemma~\ref{lem:unitdisk-approx}). Using the MIX algorithm for unit disks, introduced by Cardinal et al.~\cite{CardinalIK21}, we can maintain an independent set $S\subset \bigcup_{i=1}^4 S_i$ of size $\Omega(\max \{|S_1|, |S_2|, |S_3|, |S_4|\})$ at all times, which is a constant-factor approximation of the MIS.

Moreover, our dynamic data structure for unit disks easily generalizes to fat objects of comparable sizes in $\mathbb{R}^d$ for any constant dimension $d\in \mathbb{N}$, as explained in Section~\ref{sec:fat}. 

\begin{restatable}{theorem}{theoremFat}
\label{thm:Fat}
For every $d,f\in \mathbb{N}$ and real parameters $0<r_1<r_2$, there exists a constant $C$ with the following property: For a fully dynamic collection of $f$-fat sets in $\mathbb{R}^d$, each of size between $r_1$ and $r_2$, a $C$-approximate MIS can be maintained with worst-case update time $O(\log{n})$, and optimal output-sensitive reporting.
\end{restatable}

Our main result is a dynamic data structure for MIS over disks of arbitrary radii in the plane. 

\begin{restatable}{theorem}{theoremArbitraryDisks}
\label{thm:ArbitraryDisks}
For a fully dynamic set of disks of arbitrary radii in the plane, a constant-factor approximate maximum independent set can be maintained in polylogarithmic expected amortized update time.
\end{restatable}

We extend the core ideas developed for unit disks with several new ideas, in Section~\ref{sec:disks}. 
Specifically, we still maintain a constant number of ``grids'' such that every disk lies in one of the grid cells. For each ``grid'', we maintain an independent set $S_i$ that contains at most one disk from each cell. Then we use the MIX algorithm for disks in the plane~\cite{CardinalIK21} to maintain a single independent set $S\subset \bigcup_{i} S_i$, which is a constant-factor approximation of MIS. 

\begin{figure}[b]
    \centering
    \includegraphics{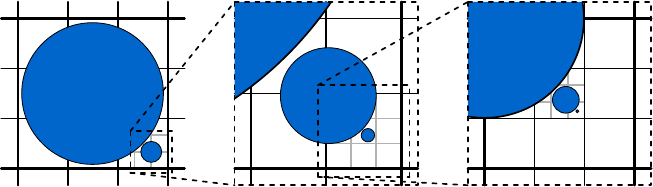}
    \caption{A nonatree with height linear in the number of stored disks, whose radii decay exponentially. By removing disks on certain intermediate levels, we can ensure that a compressed nonatree (with compressed nodes) has linear height. 
    }
    \label{fig:exponential-disks}
\end{figure}

However, we need to address several challenges that we briefly review here. 
\begin{enumerate}
    \item First, each disk should be associated with a grid cell of comparable size. This requires several scales in each shifted grid. The cells of a standard quadtree would be the standard tool for this purpose (where each cell is a square, recursively subdivided into four congruent sub-squares). Unfortunately, shifted quadtrees do not have the property that every disk lies in a cell of comparable size. Instead we subdivide each square into $3\times 3$ congruent sub-squares, and obtain a \emph{nonatree}. The crux of the proof is that 2 and  3 are relatively prime, and a shift by $\frac12$ and a subdivision by $\frac13$ are compatible (see Lemma~\ref{lem:gridpartitioning2}).
    \item For the subset of disks compatible with a nonatree, we can find an $O(1)$-approximate MIS using bottom-up tree traversal of the nonatree (using the well-known greedy strategy~\cite{MaratheBHRR95, EfratKNS00}). We can also dynamically update the greedy solution by traversing an ascending path to the root in the nonatree. 
    However, the height of the nonatree (even a compressed nonatree) may be $\Theta(n)$ for $n$ disks (see Figure~\ref{fig:exponential-disks}). In general, we cannot afford to traverse such a path in its entirety, since our update time budget is polylogarithmic. We address this challenge with the following four ideas.
    \begin{enumerate}
        \item We split each nonatree into two trees, combining alternating levels in the same tree and increasing the indegree from $3\cdot 3=9$ to $9^2=81$. This ensures that for any two disks in cells that are in ancestor-descendant relation, the radii differ by a factor of at least 3.
        \item We maintain a ``clearance'' around each disk in our independent set, in the sense that if we add a disk $d$ of radius $r$ to our independent set in a cell $c$, then we require that the  disk $3d$ (of the same center and radius $3r$) is disjoint from all larger disks that we add in any ancestor cell $c'$ of $c$. This ``clearance'' ensures that when a new disk is inserted, it intersects \emph{at most one} larger disk that is already in our independent set (Lemma~\ref{lem:clearance}).
        \item When we traverse an ascending path of the (odd or even levels of the) nonatree, we might encounter an alternating sequence of additions into and removals from the independent set: We call this a \emph{cascade sequence}. We stop each cascade sequence after a constant number of changes in our independent set and show that we still maintain a constant-factor approximation of a MIS. 
        \item Finally, when we traverse an ascending path in the  (odd or even levels of the) nonatree, we need a data structure to find the next required change: When we insert a disk $d$, we can prove that it is easy to find the next level where $d$ may intersect a larger disk in the current independent set (Lemma~\ref{lem:closest-obstacle}). However, when we delete a disk from~$S_i$, we need to find the next level where we can add another disk of the same or larger size instead. For this purpose, we use a dynamic farthest neighbor data structure by Kaplan et al.~\cite{KaplanMRSS20} (which generalizes Chan's famous dynamic convex hull data structure~\cite{DBLP:journals/jacm/Chan10,Chan20a}), 
        that supports polylogarithmic query time and polylogarithmic expected amortized update time. 
    \end{enumerate}
\end{enumerate}

One bottleneck in this framework is the farthest neighbor data structure~\cite{KaplanMRSS20,Liu22}. This provides only \emph{expected amortized} polylogarithmic update time, and it works only for families of ``nice'' objects in the plane (such as disks or homothets of a convex polygon, etc.). This is the only reason why our algorithm does not guarantee deterministic worst-case update time, and it does not extend to balls in $\mathbb{R}^d$ for $d\geq 3$, or to arbitrary fat objects in the plane. All other steps of our machinery support deterministic polylogarithmic worst-case update time, as well as balls in $\mathbb{R}^d$ for any constant dimension $d\in \mathbb{N}$, and fat objects in the plane. 

Another limitation for generalizing our framework is the MIX algorithm, which smoothly transitions from one independent set to another. Cardinal et al.~\cite{CardinalIK21} established MIX algorithms for fat objects in $\mathbb{R}^d$ for any constant $d\in \mathbb{N}$ and their proof heavily relies on separator theorems. However, they show, for example, that a sublinear MIX algorithm is impossible for rectangles in the plane.

Finally, in Section~\ref{sec:lb}, we note that, even for a dynamic set of unit disks in  the plane, it is impossible to maintain a $(1+\varepsilon)$-approximate MIS with amortized update time $n^{O((1/\varepsilon)^{1-\delta})}$ for any~$\varepsilon$, $\delta>0$, unless the Exponential Time Hypothesis (\textsf{ETH}) fails. This follows from a reduction to a result by Marx~\cite{marx2007optimality}.

\section{Preliminaries}
\label{sec:pre}

\paragraph{Fat Objects.} Intuitively, \emph{fat} objects approximate balls in $\mathbb R^d$.
Many different definitions have been used for fatness; we use the definition due to Chan~\cite{Chan03} as a MIX algorithm (described below) has been designed for fat objects using this notion of fatness. 

The \emph{size} of an object in $\mathbb{R}^d$ is the side length of its smallest enclosing axis-aligned hypercube. 
A collection of (connected) sets in $\mathbb{R}^d$ is \emph{$f$-fat} for a constant $f>0$, if in any size-$r$ hypercube $R$, one can choose $f$ points such that if any object in the collection of size at least $r$ intersects $R$, then it contains one of the chosen points. In particular, note that every size-$r$ hypercube $R$ intersects at most $f$ disjoint objects of size at least $r$ from the collection. A collection of (connected) sets in $\mathbb{R}^d$ is \emph{fat} if it is $f$-fat for some constant $f>0$. 

\paragraph{MIX Algorithm.}
A general strategy for computing an MIS is to maintain a small number of \emph{candidate} independent sets $S_1,\ldots, S_k$ with a guarantee that the largest set is a good approximation of an MIS, and each insertion and deletion incurs only constantly many changes in $S_i$ for all $i=1,\ldots , k$. To answer a query about the size of the MIS, we can simply report $\max\{|S_1|,\ldots , |S_k|\}$ in $O(k)$ time. Similarly, we can report an entire (approximate) MIS by returning a largest candidate set. However, if we need to maintain a single (approximate) MIS at all times, we need to smoothly switch from one candidate to another. 
This challenge has recently been addressed by the MIX algorithm introduced by Cardinal et al.~\cite{CardinalIK21}: 
\begin{quote}
\textbf{MIX algorithm}: The algorithm receives two independent sets $S_1$ and $S_2$ whose sizes sum to $n$ as input, and smoothly transitions from $S_1$ to $S_2$ by adding or removing one element at a time such that at all times the intermediate sets are independent sets of size at least $\min\{|S_1|, |S_2|\} - o(n)$.
\end{quote}
Cardinal et al.~\cite{CardinalIK21} constructed an $O(n \log n)$-time MIX algorithm for fat objects in $\mathbb{R}^d$, for constant dimension $d\in \mathbb{N}$. 

Assume that $\mathcal{D}$ is a fully dynamic set of disks in the plane, and we are given candidate independent sets $S_1,\ldots , S_k$ with the guarantee that $\max\{|S_1|,\ldots , |S_k|\}\geq c\,\cdot \mathrm{OPT}$ at all times, where $\mathrm{OPT}$ is the size of the MIS and $0<c\leq 1$ is a constant; further assume that the size of $S_i$, $i\in \{1,\ldots , k\}$, changes by at most a constant $u\geq 1$ for each insertion or deletion in $\mathcal{D}$. 
We wish to maintain a single approximate MIS $S$ at all times, where we are allowed to make up to $10u$ changes in $S$ for each insertion or deletion in $\mathcal{D}$.

Initially, we let $S$ be the largest candidate, say $S=S_i$. 
While $|S_i|> \frac12 \max\{|S_1|,\ldots , |S_k|\}$, we can keep $S=S_i$,
and it remains a $\frac{c}{2}$-approximation. As soon as $2|S_i|\leq |S_j|$, 
where  $|S_j|= \max\{|S_1|,\ldots , |S_k|\}$, we start switching from $S=S_i$ to $S=S_j$. Let $\alpha=|S_i|$  (hence $2\alpha\leq |S_j|\leq 2\alpha+1$) at the start of this process. 
We first apply the MIX algorithm for the current candidates $S_i$ and $S_j$, which replaces $S_i$ with $S_j$ in $O(\alpha\log \alpha)$ update time and $|S_i|+|S_j|\leq 3\alpha+1$ steps distributed over the next $\frac{\alpha}{10u}$ dynamic updates in $\mathcal{D}$, and it maintains an independent set $S_{\text{MIX}}$
of size $|S_{\text{MIX}}|\geq (1-o(1))\, \alpha$, \cite{CardinalIK21}.
 If $|S_i|\leq 3u$, we can swap $S_i$ to $S_j$ in a single step, so we may assume $|S_i|>3u$ and $|S_{\text{MIX}}|\geq (1-o(1))\, \alpha\geq \frac12\, \alpha$ for a sufficiently large constant $u$.
Note, however, that while running the MIX algorithm, the dynamic changes in $\mathcal{D}$ may include up to $\alpha/10$ deletions from $S_i\cup S_j$ and up to  up to $\alpha/10$ insertions into $S_j$. We perform any deletions from $S_i\cup S_j$ directly in $S$; create a LIFO queue for all insertions into $S_j$, and add these elements to $S$ after the completion of the MIX algorithm. That is, we switch from $S=S_i$ to $S=S_j$ in two phases: the MIX algorithm followed by adding any new elements of $S_j$ to $S$ using the LIFO queue. Recall that for each dynamic change in $\mathcal{D}$, set $S_j$ may increase by at most $u$ elements, and we are allowed to make $10u$ changes in $S$. Consequently,  both phases terminate after at most $\frac{1}{10u}\left(3\alpha+1\right)\cdot \frac{1}{1-1/10} \leq \frac{\alpha}{3}$ dynamic updates in $\mathcal{D}$.

Overall, we have $\mathrm{OPT}\leq \frac{1}{c}\,2\alpha+\frac{\alpha}{3} \leq \frac{7}{3c}\, \alpha$ at all times, and we maintain an independent set $S$ of size $|S|\geq S_{\text{MIX}}-\frac{\alpha}{3} \geq (\frac12 - \frac13)\alpha = \frac{\alpha}{6}\geq \frac{1}{6}\cdot \frac{3c}{7}\, \mathrm{OPT} = \frac{c}{14} \, \mathrm{OPT}$ at all times, and so $S$ remains a $\frac{c}{14}$-approximate MIS at all times. 
When both phases terminate, we have $S=S_j$ with $|S_j|\geq 2\alpha-\frac13\, \alpha = \frac53\,\alpha$ and $\max \{|S_1|,\ldots ,|S_k|\}\leq 2\alpha+\frac{1}{3}\alpha = \frac73\, \alpha$. That is, we have $|S_j|\geq \frac57\, \max\{|S_1|,\ldots , |S_k|\}$, which means that there is no need to switch $S_j$ to another independent set at that time.
We can summarize our result as follows. 

\begin{lemma}\label{lem:mix}
    For a collection of candidate independent sets $S_1, \ldots, S_k$, the largest of which is a $c$-approximate MIS at all times, we can dynamically maintain an $O(c)$-approximate MIS with $O(1)$ changes per update.
\end{lemma}

\paragraph{Dynamic Farthest Neighbor Data Structures.}
Given a set of functions $\mathcal{F}=\{f_1,\ldots, f_n\}$, $f_i:\mathbb{R}^2\rightarrow \mathbb{R}$ for $i=1,\ldots ,n$,  the \emph{lower envelope} of $\mathcal{F}$ is the graph of the function $L:\mathbb{R}^2\rightarrow \mathbb{R}$, $L(p)=\min \{f_i(p) \mid 1\leq i\leq n \}$.
Similarly, the \emph{upper envelope} is the graph of $U:\mathbb{R}^2\rightarrow \mathbb{R}$, $U(p)=\max \{f_i(p) \mid 1\leq i\leq n \}$.
A \emph{vertical stabbing query} with respect to the lower (resp., upper) envelope, for query point $p\in \mathbb{R}^2$, asks for the function $f_i$ such that $L(p)=f_i(p)$ (resp., $U(p)=f_i(p)$).

Given a set $\mathcal{D}$ of $n$ disks in the plane, we can use this machinery to find, for a query disk $d_q$, the disk in $\mathcal{D}$ that is closest (farthest) from $d_q$. Specifically, for each disk $d\in \mathcal{D}$ centered at $c_d$ with radius $r_d$, define the function $f_d:\mathbb{R}^2\to \mathbb{R}$, $f_d(p)=|pc_d|-r_d$. Note that $f_d(p)$ is the \emph{signed} Euclidean distance between $p\in \mathbb{R}^2$ and the disk $d$; that is, $f_d(p)=0$ if and only if $p$ is on the boundary of $d$, $f_d(p)<0$ if $p$ is in the interior of $d$, and $f_d(p)>0$ equals the Euclidean distance between $p$ and $d$ if $q$ is in the exterior of $d$. 
For a query point $p\in \mathbb{R}^2$, $L(p)=f_d(p)$ for a disk $d\in \mathcal{D}$ closest to $p$ (note that this holds even if $p$ lies in the interior of some disks $d\in \mathcal{D}$, where the Euclidean distance to $d$ is zero but $f_d(p)<0$). Similarly, we have $U(p)=f_d(p)$ for a disk $d\in \mathcal{D}$ farthest from $p$. Importantly, for a query disk $d_q$, we can find a closest (farthest) disk from $d_q$ by querying its center.
 
In the fully dynamic setting, functions are inserted and deleted to/from $\mathcal{F}$, and we wish to maintain a data structure that supports vertical stabbing queries w.r.t.\ the lower or upper envelope of $\mathcal{F}$. For linear functions $f_i$ (i.e., hyperplanes in $\mathbb{R}^3$), Chan~\cite{DBLP:journals/jacm/Chan10} devised a fully dynamic randomized data structure with polylogarithmic query time and polylogarithmic amortized expected update time; this is equivalent to a \emph{dynamic convex hull} data structure in the dual setting (with the standard point-hyperplane duality). After several incremental improvements, the current best version is a deterministic data structure for $n$ hyperplanes in $\mathbb{R}^3$ with $O(n\log n)$ preprocessing time, $O(\log^4 n)$ amortized update time, and $O(\log^2 n)$  worst-case query time~\cite{Chan20a}. 

Kaplan et al.~\cite{KaplanMRSS20} generalized Chan's data structure for dynamic sets of functions $\mathcal{F}$, where the lower (resp., upper) envelope of any $k$ functions has $O(k)$ combinatorial complexity. This includes, in particular, the signed distance functions from disks~\cite{DBLP:books/daglib/0031977}. In this case, the orthogonal projection of the lower envelope of $\mathcal{F}$ (i.e., the so-called \emph{minimization diagram}) is the Voronoi diagram of the disks. Their results is the following.

\begin{theorem}(\cite[Theorem~8.3]{KaplanMRSS20})
The lower envelope of a set of $n$ totally defined continuous bivariate functions of constant description complexity in three dimensions, such that the
lower envelope of any subset of the functions has linear complexity, can be maintained dynamically, so as to support insertions, deletions, and queries, so that each insertion takes $O(\lambda_s (\log n) \log^5 n)$ amortized expected time, each deletion takes $O(\lambda_s (\log n) \log^9n)$ amortized expected time, and each query takes $O(\log^2n)$ worst-case deterministic time, where $n$ is the number of functions currently in the data structure. The data structure requires $O(n \log^3 n)$ storage in expectation.
\end{theorem}

Subsequently, Liu~\cite[Corollary~16]{Liu22} improved the deletion time to  $O(\lambda_s (\log n) \log^7 n)$ amortized expected time.
Here $\lambda_s(t)$ is the maximum length of a 
Davenport-Schinzel sequence \cite{DBLP:books/daglib/0080837} on $t$ symbols of order $s$. For signed Euclidean distances of disks, we have $s=6$~\cite{KaplanMRSS20} and $\lambda_6(t)\ll O(t\log t)\ll O(t^2)$. For simplicity, we assume $O(\log^{9}n)$ expected amortized update time and $O(\log^2 n)$ worst-case query time. Overall, we obtain the following for disks of arbitrary radii. 

\begin{lemma}\label{lem:disjointness}
For a dynamic set $\mathcal{D}$ of $n$ disks in the plane, there is a randomized data structure that supports disk insertion in $O(\log^7 n)$ amortized expected time, disk deletion in $O(\log^{9} n)$ amortized expected time; and the following queries in $O(\log^2 n)$ worst-case time. \textbf{Disjointness query}: For a query disk $d_q$, find a disk in $\mathcal{D}$ disjoint from $d_q$, or report that all disks in $\mathcal{D}$ intersect $d_q$. 
\end{lemma}
\begin{proof}
We use the dynamic data structure in~\cite[Theorem~8.3]{KaplanMRSS20} with the update time improvements in~\cite{Liu22} for the signed Euclidean distance from the disks in $\mathcal{D}$. Given a disk $d_q$ centered at $c_q$, we can answer disjointness 
queries as follows. 
The vertical stabbing query for the upper envelope at point $c_q$ returns a disk $d\in \mathcal{D}$ farthest from $c_q$. If $d_q\cap d=\emptyset$, then return $d$, otherwise report that all disks in $\mathcal{D}$ intersect $d_q$.   
\end{proof}

We refer to the data structure in Lemma~\ref{lem:disjointness} as 
the \emph{dynamic farthest neighbor} (DFN) data structure.
We remark that Chan~\cite{Chan20a} improved the update time when the functions $\mathcal{F}=\{f_1,\ldots , f_n\}$ are distances from $n$ \emph{point sites} in the plane. De~Berg and Staals~\cite{BergS23} generalized these results to dynamic $k$-nearest neighbor data structures for $n$ point sites in the plane.

\section{Unit Disks in the Plane}
\label{sec:unit}
We first consider the case where the fully dynamic set~$\mathcal{D}$ consists of disks of the same size, namely unit disks (with radius $r=1$). Intuitively, our data structure maintains multiple grids, each with their own potential solution. For each grid, disks whose interior is disjoint from the grid lines contribute to a potential solution. We show that at any point in time, the grid that finds the largest solution holds a constant-factor approximation of MIS. 

\paragraph{Shifted Grids.} We define four axis-aligned square grids $G_1, \ldots, G_4$, in which each grid cell has side length 4. For $G_1$ the grid lines are $\{x=4i\}$ and $\{y=4i\}$ for all $i\in\mathbb{Z}$. For $G_2$ and $G_3$, respectively, the vertical and horizontal grid lines are shifted with respect to $G_1$: for $G_2$ the vertical lines are $\{x=4i+2\}$, while for $G_3$ the horizontal lines are $\{y=4i+2\}$, again for all $i\in\mathbb{Z}$. Finally, $G_4$ is both horizontally and vertically shifted, having lines $\{x=4i+2\}$ and $\{y=4i+2\}$ for all $i\in\mathbb{Z}$ (see Figure~\ref{fig:grids}a).

\begin{figure}
    \centering
    \includegraphics{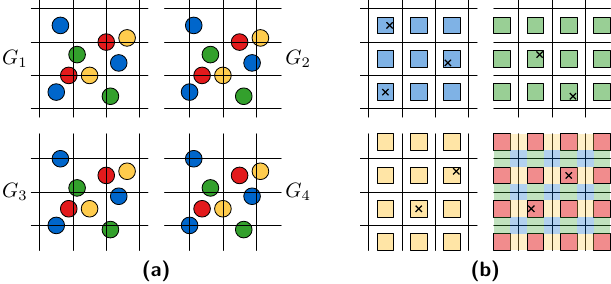}
    \caption{\textbf{\textsf{(a)}} The four shifted grids $G_1$, $G_2$, $G_3$, and $G_4$, which respectively do not intersect the blue, green, yellow, and red disks. \textbf{\textsf{(b)}} The radius-1 squares inside grid cells of the four grids, along with the center points of the disks that lie completely inside grid cells, as crosses. In the bottom right, besides red squares for $G_4$, the squares of all other grids are added to show that the squares together partition the plane.}
    \label{fig:grids}
\end{figure}

\begin{lemma}\label{lem:gridpartitioning}
    Every unit disk in $\mathbb{R}^2$ is contained in a grid cell of at least one of the shifted grid $G_1,\ldots, G_4$. Consequently, for a set~$S$ of unit disks, the cells of one of the grids jointly contain at least $|S|/4$ disks from~$S$.
\end{lemma}
\begin{proof}
    The distance between two vertical lines $\{x=4i\}$ and $\{x=4j+2\}$, for any $i, j\in \mathbb{Z}$, is at least two. A unit disk $d$ has diameter 2, so its interior cannot intersect two such lines. Consequently, the vertical strip $\{4i\leq x\leq 4i+4\}$ or $\{4i-2\leq x\leq 4i+2\}$ contains $d$ for some $i\in \mathbb{Z}$.
     Similarly, the horizontal strip $\{4j\leq x\leq 4j+4\}$ or $\{4j-2\leq x\leq 4j+2\}$ contains $d$ for some $i\in \mathbb{Z}$. The intersection of these strips is a cell in one of the grids, which contains $d$. This proves the first claim; the second claim follows from the pigeonhole principle. 
\end{proof}

Because of Lemma~\ref{lem:gridpartitioning} we know that one of the grids contains at least a constant fraction of an optimum solution $\mathrm{OPT}$, namely at least $\frac14\,|\mathrm{OPT}|$ disks.

To maintain an approximate MIS over time, we want to store information about the disks, such that we can efficiently determine the disks inside a particular grid cell, and given a disk, which grid cell(s) it is contained in. Each disk $d\in\mathcal{D}$ is represented by its center $p$ and we determine whether~$d$ is inside a cell by checking whether~$p$ is inside the $2\times 2$ square centered inside each grid cell (see Figure~\ref{fig:grids}b): Since we deal with unit disks, when a center is inside a grid cell and at least unit distance from the boundary, the corresponding disk is completely inside the grid cell. By making these $2\times 2$ square regions closed on the bottom and left, and open on the top and right, we can ensure that the union of these regions, over all cells of all four grids, partitions the plane; see Figure~\ref{fig:grids}b (bottom right). As a result, every disk is assigned to exactly one cell of exactly one grid. For each grid cell that contains at least one disk, we add an arbitrary disk to the independent set of that grid. This yields an independent set $S_i$ for each grid~$G_i$.

\begin{lemma}\label{lem:unitdisk-approx}
Let $S_1,\ldots,S_4$ be the independent sets in the set~$\mathcal{D}$ of unit disks computed for $G_1,\ldots,G_4$, respectively. The largest of $S_1,\ldots, S_4$ is a 12-approximation of a MIS for~$\mathcal{D}$.
\end{lemma}
\begin{proof}
    Let $\mathrm{OPT}\subseteq \mathcal{D}$ be a MIS. By Lemma~\ref{lem:gridpartitioning}, there is a grid $G_i$ whose cells jointly contain a subset $\mathrm{OPT}_i\subset \mathrm{OPT}$ of size $|\mathrm{OPT}_i|\geq \frac14\, |\mathrm{OPT}|$. Two unit disks are disjoint if the distance between their centers is more than 2. Consider one of the $2\times 2$ squares inside a cell of $G_i$. Recall that it is open on the top and right, and hence the $x$- or $y$-coordinates of two centers in this square differ by less than 2. Thus, at most three centers fit in a $2\times 2$ square and each grid cell can therefore contain at most three unit disks from $OPT_i$. 
    Consequently, at least $\frac13\, |OPT_i|\geq \frac{1}{12}\, |OPT|$ cells of the grid $G_i$ each contain at least one disk of $\mathcal{D}$. Thus, we return an independent set of size at least~$\frac{1}{12}\, |OPT|$.
\end{proof}

We previously considered only how to compute a constant-factor approximation of a MIS of~$\mathcal{D}$. Next we focus on how to store the centers, such that we can efficiently deal with dynamic changes to the set~$\mathcal{D}$. Our dynamic data structure consists of multiple self-balancing search trees~$T_\mathcal{D}$, and~$T_1$, $T_2$, $T_3$, and $T_4$. The former stores an identifier for each disk in~$\mathcal{D}$, while the latter four store only the identifiers of disks in $S_1,\ldots,S_4$, respectively. More precisely~$T_\mathcal{D}$ has a node for each grid cell that contains a disk, indexed by the bottom left corner $(x,y)$ of the grid cell. For each such grid cell, an additional self-balancing search tree stores all identifiers of disks in the cell. When a unit disk~$d$ is added to or deleted from~$\mathcal{D}$, we use these trees to update the sets~$S_1,\ldots,S_4$.

\begin{lemma}\label{lem:unitdisk-dynamic}
Using self-balancing search trees~$T_\mathcal{D}$ and $T_1,\ldots, T_4$, containing at most $n$ elements, we can maintain the independent sets $S_1,\ldots, S_4$ in grids $G_1,\ldots,G_4$, in~$O(\log{n})$ update time per insertion/deletion.
\end{lemma}
\begin{proof}
    For the insertion of a disk $d$ with center point~$(x,y)$, we find the bottom-left corner of the unique $2\times 2$ square~$s$ inside a cell $c\in G_i$, for $i \in \{1,2,3,4\}$, that contains $d$: Consider the grid of $2\times 2$ squares, and observe that all corners have even coordinates. We can therefore round the coordinates of the center point of~$d$ to $(\lfloor x\rfloor, \lfloor y\rfloor)$, and if either coordinate is odd, subtract one. We then find the cell~$c$ that~$d$ is contained in. We query~$T_\mathcal{D}$ with this cell and report whether the node for~$c$ exists. In case a point is found, we already have selected a disk for cell~$c$ in the potential solution~$S_i$. We therefore insert the identifier of~$d$ only in $T_\mathcal{D}$. However, if no point is found, $c$ must be empty and we insert a node for~$c$ into~$T_\mathcal{D}$ and initialize a self-balancing search tree at this node. Additionally, the independent set~$S_i$ can then grow by one by adding $d$ to this set. We hence insert the identifier of~$d$ into both the new search tree for~$c$ in $T_\mathcal{D}$ and into $T_i$.

    For the deletion of a disk $d$, we again find the bottom-left corner of the square~$s$ inside some cell~$c\in G_i$ that contains $d$, and query~$T_i$ with the identifier of~$d$. Regardless of whether we found~$d$, we now delete~$d$ from both the search tree for~$c$ in $T_\mathcal{D}$ and from $T_i$. If we found~$d$ in $T_i$, then we have to check whether we can replace it with another disk in~$c$. We do this by checking whether the search tree for~$c$ in $T_\mathcal{D}$ is empty (i.e., the pointer to root is null). If we find an identifier for a disk~$d'$ in cell~$c$ at the root, we insert it into $T_i$, so that the corresponding disk replaces~$d$ in~$S_i$. If we find no such identifier, we remove the node for~$c$ from~$T_\mathcal{D}$.

    All interactions with self-balancing search trees take~$O(\log{n})$ time. Each dynamic update requires a constant number of such operations, thus we spend~$O(\log{n})$ time per update.
\end{proof}

If we need to report the (approximate) size of the MIS, we simply report $\max\{ |S_1|,\ldots , |S_4|\}$, which is a 12-approximation. To output an (approximate) maximum independent set, we can simply choose a largest solution out of $S_1,\ldots,S_4$ and output all disks in the corresponding search tree~$T_i$ in time linear in the number of disks in this solution. Thus, Lemmata~\ref{lem:unitdisk-approx} and~\ref{lem:unitdisk-dynamic} together show that our dynamic data structure can handle dynamic changes in worst-case polylogarithmic update time, and report a solution in optimal output-sensitive time.

\theoremUnitDisk*

To explicitly maintain an independent set $S$ of size $\Omega(|OPT|)$ at all times, we can use the MIX function for unit disks~\cite{CardinalIK21} to (smoothly) switch between the sets $S_1,\ldots , S_4$. In particular, $S$ is a subset of $S_1\cup \ldots \cup S_4$, and $|S|\geq \Omega(\max\{ |S_1|,\ldots , |S_4|\})$ by Lemma~\ref{lem:mix}. Using the MIX function for unit disks~\cite{CardinalIK21}, we can hence explicitly maintain a constant-factor approximation of a MIS.

\section{Fat Objects of Comparable Size in Higher Dimensions}
\label{sec:fat}
Our algorithm to maintain an constant-factor approximation of a MIS of unit-disks readily extends to maintaining such an MIS approximation for fat objects of comparable size in any constant dimension~$d$. Remember that the size of a (fat) object is determined by the side length of its smallest enclosing (axis-aligned) hypercube. We define fat objects to be of comparable size, if the side length of their smallest enclosing (axis-aligned) hypercube is between real values $r_1$ and $r_2$.

\theoremFat*
\begin{proof}
    Similar to the unit-disk case, we define $2^d$ $d$-dimensional shifted (axis-aligned and square) grids~$G_1,\ldots,G_{2^d}$ with side length~$2\cdot r_2$: One base grid~$G_1$ and $2^d-1$ grids that (distinctly) shift the base grid in (at least one of) the $d$-dimensions by $r_2$.

    Since each object~$o$ has a size of at most~$r_2$, and grid lines defined by the union of all grids are at distance~$r_2$ from one another in every dimension, there is a grid cell in one of the grids that contains~$o$. By the pigeonhole principle, one grid~$G_i$ must therefore contain at least~$2^{-d}$ of all objects, and hence the same fraction of a MIS (analogously to Lemma~\ref{lem:gridpartitioning}).

    Furthermore, as the objects are of size at least~$r_1$, there is some constant~$c$, for which it holds that no more than~$c$ fat objects fit in a single grid cell, and observe that $c > \lfloor(\frac{r_2}{r_1})^d \rfloor$.  
    Following the unit-disk algorithm, we take a single object per grid cell in our independent set~$S_i$ of a grid~$G_i$, and hence each~$S_i$ is of size at most $\frac{1}{c}$ times the size of the MIS in~$G_i$.

    Since each~$S_i$ maintain a $c$-approximation for the MIS of the set of disks in~$G_i$, and at least one of the~$2^d$ grids holds a $2^d$-approximation of the global MIS, we get a $C$-approximation of MIS, with $C = c\cdot 2^d$ (analogously to Lemma~\ref{lem:unitdisk-approx}). 

    We again use self-balancing trees, in which we store the identifiers of the objects. We find the grid cell that contains an object by rounding the coordinates of the center point of the bounding hypercube of each object, analogously to the unit-disk case: Insertions, deletions, and reporting are handled exactly as in the unit-disk case, and hence we can prove, similarly to Lemma~\ref{lem:unitdisk-dynamic}, that insertions and deletions are handled in $O(\log{n})$ time and reporting is done in time linear in the size of the reported set~$S_i$.
\end{proof}

As in the unit-disk case, we can use the MIX function~\cite{CardinalIK21} (which works for fat objects) to switch between the independent sets~$S_1,\ldots,S_{2^d}$ of the individual grids, and hence explicitly maintain an approximate MIS at all times. We know by Lemma~\ref{lem:mix} that this results in an independent set of size~$\Omega(\max\{|S_1|,\ldots,|S_{2^d}|\})$.

\section{Disks of Arbitrary Radii in the Plane}
\label{sec:disks}

In this section, we study the DGMIS problem for a set of disks of arbitrary radii. The general idea of our new data structure is to break the set of disks $\mathcal{D}$ into subsets of disks of comparable radius. 
We will use several instances of the shifted grids $G^i_1,\ldots, G^i_4$, as we used in the unit disk case, where the grid cells have side length $3^i$, and are shifted by $\frac{3^{i}}{2}$, for $i\in \mathbb{Z}$. We say that the grids~$G^i_1,\ldots,G^i_4$ form the set~$\mathcal{G}_i$. In Section~\ref{sec:hierarchical}, we explain how hierarchical grids can be used for computing a constant-factor approximation for static instances.
Then, in Section~\ref{sec:dynamization}, we make several changes in the static data structures, to support efficient updates, while maintaining a constant factor approximation. In Section~\ref{sec:DDD}, we describe cell location data structures for our hierarchical grids and a hierarchical farthest neighbor data structure. Finally, in Section~\ref{sec:dynamic}, we stitch all these ingredients together to show how to maintain a constant-factor approximate maximum independent set in a fully dynamic setting, with expected amortized polylogarithmic update time. 

\subsection{Static Hierarchical Data Structures}
\label{sec:hierarchical}

\paragraph{Dividing Disks over Buckets.} In the grids of set~$\mathcal{G}_i$ we store disks with radius~$r$, where 
$\frac{3^{i-1}}{4} < r \leq \frac{3^{i}}{4}$. We refer to the data structures associated with one value $i$ as the \emph{bucket}~$i$. Compared to the unit disk case, where we considered only disks of radius $\frac{1}{4}$ times the side length of the grid cells, we now have to deal with disks of varying sizes even in one set~$\mathcal{G}_i$ of shifted grids. However, every disk is still completely inside at least one grid cell. To see this, observe that no two vertical or two horizontal grid lines in one grid of bucket $i$ can intersect a single disk with a radius lying in the range~$(\frac{3^{i-1}}{4}, \frac{3^{i}}{4}]$. Indeed, such disks have a diameter at most $\frac{3^{i}}{2}$, while grid lines are at least $3^i$ apart.

Furthermore, our choice for side length $3^i$ for bucket $i$ was not arbitrary: Consider also adjacent bucket $i-1$ and observe that each cell $c$ of grid~$G_1^i$ is further subdivided into nine cells of grid~$G_1^{i-1}$, in a $3\times 3$ formation. We say that $c$ is \emph{aligned} with the nine cells in bucket $i-1$. We define the same parent-child relations as in a quadtree: If a grid cell~$c$ in a lower bucket is inside a cell~$c_p$ of an adjacent higher bucket, we say that~$c$ is a child (cell) of~$c_p$, or that $c_p$ is the parent (cell) of~$c$. In general, we write $c_1\prec c_2$ if cell $c_1$ is a descendant of cell $c_2$; $c_1\preceq c_2$ if equality is allowed. We call the resulting structure a \emph{nonatree}, and we will refer to the nonatree that relates all grids $G^j_1$ as $N_1$. In Figure~\ref{fig:nonatree}a we illustrate the grids of two consecutive buckets in a nonatree.

Crucially, all grids $G^j_2$ also align, and the same holds for $G^j_3$ and $G^j_4$. This happens because horizontally and vertically, grid cells are subdivided into an odd number of cells (three in our case), and the shifted grids are displaced by half the side length of the grid cells. Thus, for $G^j_2$ and $G^j_4$, the horizontal shift in buckets~$i$ and~$i-1$ ensures that every third vertical grid line of bucket~$i-1$ aligns with a vertical grid line of bucket~$i$. The exact same happens for the horizontal grid lines of $G^j_3$ and $G^j_4$, due to the vertical shift. Thus, the horizontally shifted grids also form a nonatree~$N_2$, and similarly, we define $N_3$ and~$N_4$. 

For each bucket~$i$, we maintain the four self-balancing search trees.
Let $\mathcal{D}_i \subseteq \mathcal{D}$ be the subset of disks stored in~$\mathcal{G}_i$ and let $S_1,\ldots ,S_4$ be an independent set in $G^i_1,\ldots ,G^i_4$, then we maintain in $T^i_{\mathcal{D}}$ all disks in $\mathcal{D}_i$ and in $T^i_1,\ldots,T^i_4$ the disks in $S_1,\ldots ,S_4$, similar to the unit disk case. 

\begin{figure}
    \centering
    \includegraphics{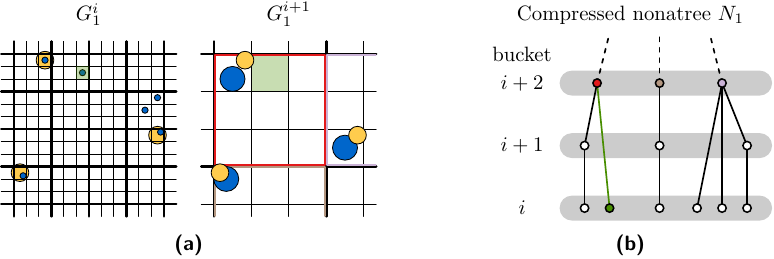}
    \caption{\textbf{\textsf{(a)}} Two compatible grids in buckets $i$ and $i+1$, with (blue) disks of~$D$ in relevant cells. In particular, the green cell in $G_1^i$ is relevant, but its (green) parent cell in~$G_1^{i+1}$ is not. Three (yellow) obstacle disks of~$G_1^i$ are drawn in both grids. Only one blue disk in~$G_1^{i+1}$ is disjoint from an obstacle, and can be chosen in the greedy bottom-up strategy. \textbf{\textsf{(b)}} Part of the compressed nonatree~$N_1$ corresponding to \textbf{\textsf{(a)}}: The colored nodes of bucket~$i+2$ correspond to colored squares in~\textbf{\textsf{(a)}} of the same color. Because the green cell in~$G_1^{i+1}$ is not relevant, and does not have relevant children in two subtrees, it is not represented in~$N_1$. Instead, the green node, corresponding to the green relevant cell in~$G_1^i$, directly connects to an ancestor in bucket~$i+2$ (by the green edge).}
    \label{fig:nonatree}
\end{figure}

\paragraph{Approximating a Maximum Independent Set.} We will now use the data structures to compute an approximate MIS for disks with arbitrary radii. Note that, we defined buckets for $i\in \mathbb{Z}$, but we will use only those buckets that store any disks, which we call \emph{relevant} buckets. Within these buckets, we call grid cells that contain disks the \emph{relevant} grid cells. Figure~\ref{fig:nonatree} illustrates the concepts introduced in this paragraph and the upcoming paragraphs.

Let $B$ be the sequence of relevant buckets, ordered on their parameter~$i$. To compute a solution, we will consider the buckets in~$B$ in ascending order, starting from the lowest bucket, which holds the smallest disk, and has grids with the smallest side length, up to the highest bucket with the largest disks, and largest side lengths. We follow a greedy bottom-up strategy for finding a constant-factor approximation of an MIS of disks. To prevent computational overhead in this approach, our nonatrees are \emph{compressed}, similar to compressed quadtrees~\cite[Chapter~2]{har2011geometric}: Each nonatree consists of a root cell, all relevant grid cells, and all cells that have relevant grid cells in at least two subtrees. As such, each (non-root) internal cell of our nonatrees either contains a disk, or merges at least two subtrees that contain disks, and hence the total number of cells in a compressed nonatree is linear in the number of disks it stores, which is upper bounded by~$O(n)$. 

Specifically, two high-level steps can be distinguished in our approach:

\begin{enumerate}
    \item In the lowest relevant bucket, we simply select an arbitrary disk from each relevant grid cell. In other relevant buckets, we consider for each grid cell~$c\in G^i_k$ the subdivision of $c$ in $G^{j}_k$ in the preceding relevant bucket~$j<i$. We try to combine the independent set from the relevant child(ren) of~$c$ with at most one additional disk in $c$. To communicate upwards which disks has been included in our independent set, we use \emph{obstacle disks} (these are not necessarily input disks; see the next step). Once all relevant cells have been handled, we output the largest independent set among the four sets computed for the shifted nonatrees $N_1,\ldots, N_4$. This produces a constant-factor approximation, as shown in Lemmata~\ref{lem:gridpartitioning2}--\ref{lem:greedy}.
    \item The obstacle disk in the previous step may cover more area than the disks in the independent set of the children of~$c$. Hence, we consider computing the obstacle disk only for independent sets originating from a single child cell. In this case, we choose as the obstacle the smallest disk covering the contributing child cell in question. The obstacle will then be of comparable size to that child cell, and hence also comparable to the contributed disk, intersecting at most a constant number of disks in the parent cell~$c$. Otherwise, if the independent set of the children originates from more than one child, we simply do not add a disk from~$c$, even if that may be possible.
    Lemmata~\ref{lem:skip-cells} and~\ref{lem:obstacle-packing} show that we still obtain a constant-factor approximate MIS under these constraints.
\end{enumerate}

We will now elaborate on the high-level steps, and provide a sequence of lemmata that can be combined to prove the approximation ratio of the computed independent set.

In the first step, we deviate from an optimal solution in three ways: We follow a greedy bottom-up approach, we take at most one disk per grid cell, and we do not combine the solutions of the shifted nonatrees. Focusing on the latter concern first, we extend Lemma~\ref{lem:gridpartitioning} to prove the same bound for our shifted nonatrees. Before we can prove this lemma, we first define the intersection between a disk and a nonatree, as follows. We say that a disk~$d$ intersects (the grid lines of) a nonatree~$N_k$, if and only if its radius~$r_d$ is in the range~$(\frac{3^{i-1}}{4}, \frac{3^{i}}{4}]$ and it intersects grid lines of~$G^i_k$.

\begin{lemma}\label{lem:gridpartitioning2}
    For a set~$S$ of disks in~$\mathbb{R}^2$, the grid lines of at least one nonatree, out of the shifted nonatrees $N_1,\ldots, N_4$, do not intersect at least $|S|/4$ disks.
\end{lemma}
\begin{proof}
    Consider the subset $D_1\subseteq S$ of disks intersecting $N_1$, If $|D_1|<\frac{3|S|}{4}$ then at least $|S|/4$ disks are not intersected by~$N_1$, and the lemma trivially holds.

    Now assume that $|D_1|\geq\frac{3|S|}{4}$ and consider the partitioning of $D_1$ into $D_2\subseteq D_1$, $D_3\subseteq D_1$, and $D_4\subseteq D_1$ which respectively intersect only vertical lines, only horizontal lines, or both vertical and horizontal lines of $N_1$. By definition of the grids that make up the nonatrees $N_2,N_3,N_4$, the disks in $D_2$ do not intersect~$N_2$, and similarly $D_3$ and $D_4$ do not intersect $N_3$ and $N_4$, respectively. Let $D^*$ be the largest set out of $D_2$, $D_3$, and $D_4$. Since $|D_1| = |D_2| + |D_3| + |D_4|$ and $|D_1|\geq\frac{3|S|}{4}$, $D^*$ must have size at least $|S|/4$. Hence, the nonatree corresponding to $D^*$ does not intersect at least $|S|/4$ disks in~$S$.
\end{proof}

Similarly, we can generalize Lemma~\ref{lem:unitdisk-approx} to work for the newly defined grids in~$\mathcal{G}_i$, that is,
for disks with different radii in a certain range. We show that taking only a single disk per grid cell into our solution is a $35$-approximation of a MIS.

\begin{lemma}\label{lem:generaldisk-cellpacking-approx}
    If $S$ is a maximum independent set of the disks in a grid cell of a nonatree~$N_k$, then $|S|\leq 35$. 
\end{lemma}
\begin{proof}
    Since we store disks of smaller radius compared to the unit disk case, the largest independent set inside a single grid cell increases from three to $35$: In a bucket~$i$ the disks have radius $r>\frac{3^{i-1}}{4}$ and the grid cells have side length $3^i$. The grid cells are therefore just too small to fit $3\cdot 4=12$ times the smallest disk radius horizontally or vertically. Hence, we cannot fit a grid of $6\times 6 = 36$ disjoint disks in one grid cell (which is the tightest packing for a square with side length $3^i$ and disks with radius $r=\frac{3^{i-1}}{4}$~\cite{wengerodt1987dichteste}; see also~\cite{Szabo2007}). 
\end{proof}

To round out the first step, we prove that our greedy strategy contributes at most a factor~$5$ to our approximation factor.

\begin{lemma}\label{lem:greedy}
    Let~$S$ be a maximum independent set of the disks in a nonatree~$N_k$ such that each grid cell in~$N_k$ contributes at most one disk. An algorithm that considers the grid cells in $N_k$ in bottom-up fashion, and computes an independent set~$S'$ by greedily adding at most one non-overlapping disk per grid cell to~$S'$, is a $5$-approximation of~$S$.
\end{lemma}
\begin{proof}
    Every disk~$d$ can intersect at most 5 pairwise disjoint disks, that have a radius at least as large as the radius of~$d$. Thus, a greedily selected disk $d\in S' \setminus S$ can overlap with at most five larger disks in $S$. These five disks are necessarily located in higher buckets (or one disk can be located in the same cell as~$d$), since all grid cells of one bucket in~$N_k$ are disjoint, and lower buckets contain smaller disks. As such, the greedy algorithm will not find these five disks before considering $d$, and cannot add them after greedily adding $d$ to $S'$. Thus, $S'$ is a $5$-approximation of $S$.
\end{proof}

For the second step, we use several data structures and algorithmic steps that help us achieve polylogarithmic update and query times in the dynamic setting. For now we analyze solely the approximation factor incurred by these techniques. We start by analyzing the approximation ratio for not taking any disk from a cell~$c$, if several of its children contribute disks to the computed independent set.

\begin{lemma}\label{lem:skip-cells}
    Let~$S$ be a maximum independent set of the disks in a nonatree~$N_k$ such that each grid cell in~$N_k$ contributes at most one disk. The independent set~$S'\subset S$, that contains all disks in $S$ except disks from cells that have two relevant child cells, is a 2-approximation of~$S$.
\end{lemma}
\begin{proof}
    Consider the tree structure~$T$ of nonatree~$N_k$. Every cell that is a leaf of~$T$ contributes its smallest disk to both~$S$ and~$S'$. Contract every edge of~$T$, that connects a cell that does not contribute a disk to~$S$, to its parent. The remaining structure~$T'$ is still a tree, and every node of the tree corresponds to a cell that contributes exactly one disk to~$S$, and hence $|S|=|T'|$. Internal nodes of~$T'$ either have two children, in which case they do not contribute a disk to~$S'$, or they have one child, in which case they do contribute a disk to~$S'$. 
    
    If we add an additional leaf to every internal node of~$T'$ that has only one child, then we get a tree~$T'_2$, where every internal node has at least two children, and every leaf corresponds to a disk in~$S'$: For internal nodes that contribute a disk, the newly added leaf corresponds to the contributed disk. Since every node has at least two children, the number of leaves of~$T'_2$ is strictly larger than~$|T'_2|/2$. It follows  that $|S'|>|T'_2|/2$.
    
    Finally, the size of $T'_2$ is at least as large as $T'$, meaning $|T'_2| \geq |T'|$. The approximation ratio of $S'$ compared to $S$ is then $\frac{|S'|}{|S|} > \frac{|T'_2|/2}{|T'|} \geq \frac{1}{2}$.
\end{proof}

Next we consider the obstacle disk that we compute when only one child cell contributes disks to the independent set. Before we elaborate on the approximation ratio of this algorithmic procedure, we first explain the steps in more detail.

For the leaf cells of a nonatree, it is unnecessary to compute an obstacle disk, since these cells contribute at most a single disk, which can act as its own obstacle disk. For a cell~$c$ that is an internal node of the nonatree, with at most one relevant child that contributes to the independent set, we have two options for the obstacle disk of~$c$. We use the obstacle disk of the child cell to determine whether there is a disk in~$c$ disjoint from the child obstacle, to either find a disjoint disk~$d$ or not. If we find such a disk~$d$, we compute a new obstacle disk for $c$, by taking the smallest enclosing disk of~$c$. If there is no such disk~$d$, then we use the obstacle disk of the child as the obstacle disk for~$c$. This ensures that the obstacle disk does not grow unnecessarily, which is relevant when proving the following approximation factor.

\begin{lemma}\label{lem:obstacle-packing}
    Let~$c$ be a cell in bucket~$i$ of nonatree~$N_k$ that contributes a disk to an independent set. The computed obstacle disk~$d_o$ can overlap with no more than $23$ pairwise disjoint disks in higher buckets.
\end{lemma}
\begin{proof}
    Let~$c$ be an obstacle cell at level $i\in \mathbb{Z}$. Then $c$ has side length $3^i$, and the disk associated with $c$ has radii in the range $(\frac14\, 3^{i-1},\frac14\, 3^i]$. 
    Obstacle disk~$d_o$ therefore has a radius of $r=\frac{\sqrt{2}}{2}\cdot 3^i$, and $\mathrm{area}(d_o)=\pi r^2=\frac{\pi}{2} \, 3^{2i}$ (see Figure~\ref{fig:obstacle-packing}a). 
 
    The radius of any disk $d\in \mathcal{D}_k$ in a higher bucket is at least $\frac14\, 3^i$. Let $\mathcal{A}$ be the set of pairwise disjoint disks in $\mathcal{D}_k$ in higher buckets that intersect $d_o$. Scale down every disk $d\in \mathcal{D}$ from a point in $d\cap d_o$ to a disk $\widehat{d}$ of radius $\widehat{r}=\frac14\, 3^{i}$; and let $\widehat{\mathcal{A}}$ be the set of resulting disks. Note that $|\widehat{\mathcal{A}}|=|\mathcal{A}|$, the disks in $\widehat{\mathcal{A}}$ are pairwise disjoint, and they all intersect $d_o$. Let $D$ be a disk concentric with $d_o$, of radius $R=3^{i}/\sqrt{2}+\frac12\, 3^{i}=\frac{1+\sqrt{2}}{2}\, 3^{i}$. 
    By the triangle inequality, $D$ contains all (scaled) disks in $\widehat{\mathcal{A}}$ (see Figure~\ref{fig:obstacle-packing}b). 
    Since the disks in $\widehat{\mathcal{A}}$ are pairwise disjoint, then $|\widehat{\mathcal{A}}|\cdot \pi\widehat{r}^2=\sum_{\widehat{d}\in\widehat{\mathcal{A}}} \mathrm{area}(\widehat{d}) \leq \mathrm{area}(D)=\pi R^2$. 
    This yields $|\mathcal{A}|=|\widehat{\mathcal{A}}| \leq \mathrm{area}(D)/\mathrm{area}(\widehat{d}) = R^2/\widehat{r}^2 = (2+2\sqrt{2})^2\approx 23.31$, as claimed.
\end{proof}

\begin{figure}
    \centering
    \includegraphics{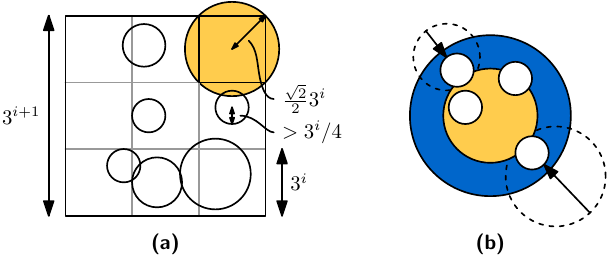}
    \caption{\textsf{\textbf{(a)}} A yellow obstacle disk for a cell in bucket~$i$ along with disks in bucket~$i+1$. The grid lines for bucket~$i$ are drawn in grey, except for the cell with the obstacle disk. \textsf{\textbf{(b)}} The dashed disks of radius larger than~$3^i/4$ are scaled down such that all white disks have radius~$3^i/4$ and intersect the yellow obstacle disk~$d_o$. All white disks are contained in the blue disk~$D$ with a radius $3^i/2$ larger than~$d_o$.
    }
    \label{fig:obstacle-packing}
\end{figure}

\begin{lemma}\label{lem:14700}
For a set of disks in the plane, one of our shifted nonatrees~$N_1,\ldots,N_4$ maintains an independent set of size $\Omega(|\mathrm{OPT}|)$, where $\mathrm{OPT}$ is a MIS.
\end{lemma}
\begin{proof}
    By Lemma~\ref{lem:gridpartitioning2} we know that at least $|\mathrm{OPT}|/4$ disks of $\mathrm{OPT}$ are stored in one of the four nonatrees, say~$N_{\mathrm{OPT}}$. Lemma~\ref{lem:generaldisk-cellpacking-approx} tells us that at most 35 disks in $\mathrm{OPT}$ can be together in a single cell of such a nonatree. Since the maintained independent set takes at most a single disk from each cell, it is at least a $4\cdot 35 = 140$ approximation of~$\mathrm{OPT}$.
    
    By considering the cells in bottom-up fashion when constructing the independent set, Lemma~\ref{lem:greedy} shows that a $5$-approximation of the $140$-approximation will be found, leading to an approximation factor of $5\cdot 140 = 700$. Lemma~\ref{lem:skip-cells} allows us to remove those disks in cells that have two relevant child cells, to find a $2$-approximation of the independent set before removing the disks, leading to a $2\cdot 700 = 1400$ approximation. 
    
    Finally, we use an obstacle disk, instead of the actual disks in the independent set of a child cell to check for overlap with disks in the parent cell. Lemma~\ref{lem:obstacle-packing} tells us that we disregard at most 23 disks in higher buckets for overlapping with the obstacle. Since it is unclear whether these 23 disks are really overlapping with disks in the independent set of the child cell, and since the obstacle disk is computed only when a child contributes at least 2 disks to the independent set, this leads to an approximation factor of $\frac{25}{2}$. The maintained solution in~$N_{OPT}$ is hence a $\frac{25}{2}$-approximation of a $1400$-approximation.
\end{proof}

\subsection{Modifications to Support Dynamic Maintenance}\label{sec:dynamization}

In Section~\ref{sec:hierarchical}, we defined four hierarchical grids (nonatrees) $N_1,\ldots , N_4$, described a greedy algorithm that computes independent sets $S_1,\ldots , S_4$ that are consistent with the grids, and showed that a largest of the four independent sets is a constant-factor approximation of the MIS. 

In this section, we make several changes in the static data structures, to support efficient updates, while maintaining a constant-factor approximation. Then in Section~\ref{sec:dynamic}, we show that the modified data structures can be maintained dynamically in expected amortized polylogarithmic update time. We start with a summary of the  modifications:

\begin{itemize}
    \item \textbf{Sparsification.} We split each nonatree $N_i$, $i\in \{1,\ldots, 4\}$, into two trees $N_i^{\rm odd}$ and $N_i^{\rm even}$, one containing the odd levels and the other containing the even levels. As a result, the radii of disks at different (non-empty) levels differ by at least a factor of 3.   
    \item \textbf{Clearance.} For a disk $d$ of radius $r$, let $3d$ denote the concentric disk of radius $3r$. Recall that our greedy strategy adds disks to an independent set $S$ in a bottom-up traversal of a nonatree. When we add a disk $d\in \mathcal{D}$ to $S$, we require that we do not add any larger disk to $S$ that intersects $3d$. In particular, we will use obstacle disks of the form $3d'$, where $d'$ is the smallest enclosing disk of a cell.
    A simple volume argument shows that this modification still yields a constant-factor approximation. As a result, if a new disk is inserted, it intersects at most one larger disk in $S$, which simplifies the update operation in Section~\ref{sec:dynamic}.
    \item \textbf{Obstacle Disks and Obstacle Cells.} In Section~\ref{sec:hierarchical}, we defined obstacle disks for the disks in $S_k$. To support dynamic updates, we use slightly larger obstacle disks, to implement the clearance in our data structures. These obstacle disks are associated with cells of the nonatree $N_k$, which are called obstacle cells (true obstacles). Cells of the nonatree with two or more children are also considered as obstacle cells (merge obstacles), thus the obstacle cells decompose each nonatree into ascending paths. 
    \item \textbf{Barrier Disks.} The na\"{\i}ve approach for a dynamic update of the independent set $S_k$ in a nonatree $N_k$ would work as follows: When a new disk $d$ is inserted or deleted, we find a nonatree $N_k$ and a cell $c\in N_k$ associated with $d$; and then in an ascending path of $N_k$ from $c$ to the root, we re-compute the disks in $S_k$ associated with the cells.  Unfortunately, the height of the nonatree may be linear, and we cannot afford to traverse an ascending path from $c$ to the root. Instead, we run the greedy process only locally, on an ascending paths of $N_k$ between two cells $c_1\prec  c_2$ that contain disks $s_1,s_2\in S_k$, respectively. The greedy process guarantees that new disks added to $S_k$ are disjoint from any smaller disk in $S_k$, including $s_1$. However, the new disks might intersect the larger disk $s_2\in S_k$. In this case, we remove $s_2$ from $S$, keep it as a ''placeholder'' in a set $B_k$ of \emph{barrier disks}, and ensure that $S_k\cup B_k$ remains a dominating set of the disks of $\mathcal{D}$ contained in $N_k$. 
\end{itemize}

\paragraph{Sparsification.} Recall that for a set $\mathcal{D}$ of $n$ disks, $\mathcal{D}_i$ denoted the subset of disks of radius $r$, where $\frac{3^{i-1}}{4}<r\leq \frac{3^i}{4}$, for all $i\in \mathbb{Z}$. Let $N_1,\ldots , N_4$, be the four nonatrees defined in Section~\ref{sec:hierarchical}. For every $k\in
\{1,\ldots , 4\}$, we create two copies of $N_k$, denoted  $N_k^{\rm even}$ and $N_k^{\rm odd}$. For $i$ even (resp., odd), we associate the disks in $\mathcal{D}_i$ to the nonatrees $N_k^{\rm even}$ (resp., $N_k^{\rm odd}$). 
For simplicity, we denote the eight nonatrees $N_k^{\rm odd}$ and $N_k^{\rm even}$ as $N_1,\ldots , N_8$. We state a simple corollary to Lemma~\ref{lem:14700}.

\begin{lemma}\label{lem:14700+}
For a set of disks in the plane, one of our shifted nonatrees~$N_1,\ldots,N_8$ maintains an independent set of size $|\mathrm{OPT}|/C$, where $\mathrm{OPT}$ is a MIS and $C$ is an absolute constant.
\end{lemma}
\begin{proof}
Let $S\subset \mathcal{D}$ be a MIS of a set $\mathcal{D}$ of disks. We can partition $\mathcal{D}$ into  $\mathcal{D}^{\rm even}=\bigcup_{i\mbox{ \rm even}} \mathcal{D}_i$ and $\mathcal{D}^{\rm odd}=\bigcup_{i\mbox{ \rm odd}}  \mathcal{D}_i$. Let $S^{\rm even}= S\cap \mathcal{D}^{\rm even}$ and $S^{\rm odd}=S\cap \mathcal{D}^{\rm odd}$. 
Clearly, $|S^{\rm even}|\geq \mathrm{OPT}^{\rm even}$, $|S^{\rm odd}|\geq \mathrm{OPT}^{\rm odd}$, and $\max\{|S^{\rm even}|,|S^{\rm odd}|\}\geq \frac12\, |S| =\frac12\,  \mathrm{OPT}$. Now Lemma~\ref{lem:14700} completes the proof.   
\end{proof}

The advantage of partitioning the nonatrees into odd and even levels is the following. 

\begin{lemma}\label{lem:sparse}
Let $d_1,d_2\in \mathcal{D}$ be disks of radii $r_1, r_2>0$, respectively, associated with cells $c_1$ and $c_2$ in a nonatree $N_k$, $k\in \{1,\ldots ,8\}$. If $c_1\prec c_2$, then $3\, r_1< r_2$. 
\end{lemma}
\begin{proof}
 By construction, $N_k$ contains disks at either odd or even levels.
 Then $3^{i-1}/4< r_1\leq 3^i/4$ and $3^{i'-1}/4< r_2\leq 3^{i'}/4$ for some integers $i<i'$ of the same parity. Since $i$ and $i'$ have the same parity, then $i+2\leq i'$,
 which gives $r_1\leq 3^i/4<3^{i+1}/4<r_2$, hence $r_2/r_1>3$.
\end{proof}

\paragraph{Clearance.} The guiding principle of the greedy strategy is that if we add a disk $d$ to the independent set, we exclude all larger disks that intersect $d$. For our dynamic algorithm, we wish to maintain a stronger property: 
\begin{definition}
  Let~$S$ be an independent set of the disks in a nonatree~$N_k$ such that each grid cells in~$N_k$ contributes at most one disk.
  For $\lambda\geq 1$, we say that $S$ has \textbf{$\lambda$-clearance} if the following holds: 
  If $d_1,d_2\in S$ are associated with cells $c_1$ and $c_2$, resp., and $c_1\prec c_2$, then  
  $d_2$ is disjoint from $\lambda\,d_1'$, where $d_1'$ is the smallest enclosing disk of $c_1$. 
\end{definition}
Note that $d_1\subset d_1'$ and $\lambda\, d_1\subset \lambda\, d_1'$, and in particular $\lambda$-clearance implies that $d_2$ is disjoint from $\lambda\, d_1$. This weaker property suffices for some of our proofs (e.g., Lemma~\ref{lem:clearance}).
An easy volume argument shows that a modified greedy algorithm that maintains 3-clearance still returns a constant-factor approximate MIS (Lemma~\ref{lem:obstacle-packing2}). The key advantage of an independent set with 3-clearance is the following property, which will be helpful for our dynamic algorithm:  

\begin{lemma}\label{lem:clearance}
  Let~$S$ be an independent set of the disks in a nonatree~$N_k$ such that each grid cells in~$N_k$ contributes at most one disk; and assume that $S$ has 3-clearance. 
  Then every disk that lies in a cell in $N_k$ intersects at most one larger disk in $S$.
\end{lemma}
\begin{proof}
Let $d_0$ be an arbitrary disk in a cell $c_0$ of $N_k$, and assume that $d_0$ intersects two or more larger disks in $S$. Let $d_1$ and $d_2$ be the smallest and the 2nd smallest disks in $S$ that (i) intersect $d_0$ (ii) and are larger than $d_0$. Clearly, if $d_1$ and $d_2$ are associated with cells $c_1$ and $c_2$, then we have $c_0\prec c_1\prec c_2$. Since $d_0$ intersects the larger disk $d_1$, then $d_0\subset 3d_1$. Since $S$ has 3-clearance, then $d_2$ is disjoint from $3d_1$ (see Figure~\ref{fig:closest-obstacle}a). Consequently, $d_2$ cannot intersect $d_0$: a contradiction completing the proof.    
\end{proof}

\begin{figure}
    \centering
    \includegraphics{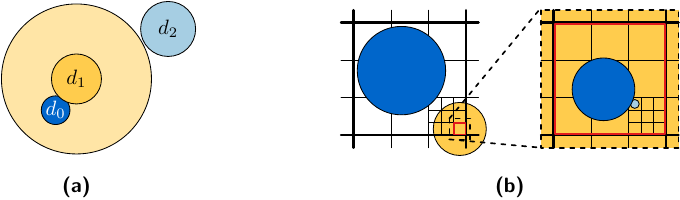}
    \caption{Constructions for Lemmata~\ref{lem:clearance} and~\ref{lem:closest-obstacle}, respectively: \textbf{\textsf{(a)}} The light yellow disk representing $3d_1$ is disjoint from~$d_2$ because of 3-clearance. \textbf{\textsf{(b)}} The light blue disk can intersect only a disk of~$S_k$ in the red cell~$c_o$; larger disks in~$S_k$ are disjoint from the yellow obstacle disk defined by~$c_o$.}
    \label{fig:closest-obstacle}
\end{figure}

\paragraph{Obstacle Cells: Decomposing a Nonatree into Ascending Paths.}
A cell $c\in N_k$ is an \emph{obstacle cell} if it is associated with a disk in $S_k$ (a \emph{true obstacle}), or it has at least two children that each contain a disk in $S_k$ (a \emph{merge obstacle}). 
For every obstacle cell $c$, we define an \emph{obstacle disk} as $o(c) = 3d'$, where $d'$ is the smallest enclosing disk of the cell $c$.
The obstacle cells decompose the nonatree into ascending paths in which each cell has relevant descendants in only a single subtree (see Figure~\ref{fig:obs-struct}a). Inside an ascending path, disks either intersect the obstacle disk of the (closest) obstacle cell below them, or are part of~$S_k$ and therefore define a true obstacle cell (see Figure~\ref{fig:obs-struct}b and~\ref{fig:obs-struct}c). 

\begin{figure}
    \centering
    \includegraphics[page=1]{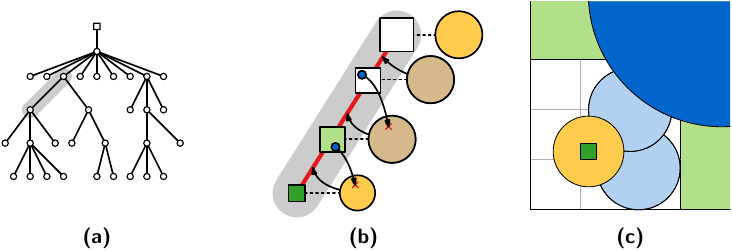}
    \caption{\textsf{\textbf{(a)}} Decomposition of a nonatree into ascending paths between merge obstacle cells. Only relevant leaves are drawn and hence all leaves are obstacle cells (disks) as well. The (square) root is not necessarily an obstacle cell. One ascending path between merge nodes is highlighted in grey. 
    The structure of the highlighted ascending path is shown \textbf{\textsf{(b)}} abstractly and \textbf{\textsf{(c)}} geometrically: The merge obstacle cells at the top and bottom (with yellow obstacle disks) each have no disk of~$S_k$ associated with them. Every other obstacle cell on the path also defines a brown obstacle disk. Each such cell contains a (dark blue) disk of~$S_k$, which is disjoint from the (closest) obstacle disk below it (indicated by red crosses). All (light blue) disks on the (red) ascending path above an obstacle cell are intersected by the obstacle below. Green colors identify cells between \textbf{\textsf{(b)}} and \textbf{\textsf{(c)}}.
    }
    \label{fig:obs-struct}
\end{figure}

\paragraph{Barrier Disks.} 
For a set of disks $\mathcal{D}$, we will maintain an independent set $S\subset \mathcal{D}$, and a set $B\subset \mathcal{D}$ of \emph{barrier disks}. 
When a disk $d$ associated with a cell $c\in N_k$ is inserted or deleted from $\mathcal{D}$, we re-run the greedy process on the nonatree locally, between the cells $c_1\preceq c\prec c_2$ that contain disks $s_1,s_2\in S$. If any of the new disks added to $S$ intersects $s_2$, then we remove $s_2$ from $S$, and add it to $B$ as a \emph{barrier disk}. Such a barrier disk defines a \emph{barrier clearance disk} $o(c_b)=3d_b$, where $d_b$ is the smallest enclosing disk of the \emph{barrier cell}~$c_b$ containing~$s_2$. This obstacle disk also implements the clearance (defined above), to guarantee that the new disks added to $S$ in this process do not intersect any disk in $S$ larger than $s_2$.  
Importantly, we maintain the properties that
(i) the obstacle disks, for all obstacle cells and barrier cells, form a dominating set for $\mathcal{D}$, that is, all disks in $\mathcal{D}$ intersect an obstacle disk of some obstacle cell or the barrier clearance disk of a barrier cell; and (ii) on any ascending path there is always an obstacle cell between two barrier cells. 

The latter property ensures that $|B|\leq 2\,|S|$ and is maintained as follows. We maintain an assignment~$\beta$ between barrier disks and the closest obstacle cells below them.
Each barrier disk $\beta(c_1)$ lies in one of the cells of the nonatree along an ascending path between two obstacle cells $c_1\prec c_2$. Each path contains at most one barrier disk.
To maintain this property after each insertion or deletion, we re-run the greedy algorithm locally on the ascending path affected by the dynamic change, and possibly continue the greedy algorithm one ascending path just above. 
Details are given in Section~\ref{sec:dynamic}. 

\paragraph{Invariants.} We are now ready to formulate invariants that guarantee that one of eight possible independent sets is a constant-factor approximation of MIS. In Section~\ref{sec:dynamic}, we show how to maintain these independent sets and the invariants in polylogarithmic time. 

For a set of disks $\mathcal{D}$, we maintain eight nonatrees $N_1,\ldots N_8$, and for each $k\in \{1,\ldots , 8\}$ we maintain two sets of disks $S_k$ and $B_k$, that satisfy the following invariants. 
\begin{enumerate}
    \item\label{inv:1} Every disk $d\in \mathcal{D}$ is associated with a cell of at least one nonatree $N_k$, $k\in \{1,\ldots , 8\}$. 
    \item\label{inv:2} In each nonatree $N_k$, only odd or only even levels are associated with disks in $\mathcal{D}$. Let $\mathcal{D}_k$ be the set of disks associated with the cells in $N_k$. 
    \item\label{inv:3} For every $k\in \{1,\ldots , 8\}$, 
         \begin{enumerate}
        \item\label{inv:3a} $S_k$ and $B_k$ are disjoint subsets of $\mathcal{D}_k$;
        \item\label{inv:3c} $S_k\subset \mathcal{D}_k$ is an independent set with 3-clearance; and
        \item\label{inv:3b} each cell of $N_k$ contributes at most one disk in $B_k\cup S_k$.
        \end{enumerate}
    \item\label{inv:obs} For every $k\in \{1,\ldots , 8\}$, 
        \begin{enumerate} 
        \item\label{inv:obsa} a cell $c\in N_k$ is an obstacle cell if it is associated with a disk in $S_k$ (a true obstacle), or it has at least two children that each contain a disk in $S_k$ (a merge obstacle). 
        \item\label{inv:obsb} For every obstacle cell $c$, we maintain a obstacle disk as $o(c) = 3d'$, where $d'$ is the smallest enclosing disk of the cell $c$.
        \end{enumerate}
    \item\label{inv:4} 
         For every $b\in B_k$,
        \begin{enumerate} 
        \item\label{inv:4c} there is a unique obstacle cell $c_o\in N_k$ such that $c_o\prec c_b$, where $c_b$ is the cell in $N_k$ associated with $b$, and the cells $c$, $c_o\prec c\prec c_b$, are neither obstacle cells nor associated with any disk in $B_k$; 
        we use the notation $\beta(c_d)=b$ to assign barrier disks to cells and $\beta(c)=\textsc{nil}$ if a cells is not assigned to any barrier disk in $B_k$; 
       \item\label{inv:bobs} For every cell $c$ associated with a barrier in $B_k$, we maintain a barrier clearance disk $o(c) = 3d'$, where $d'$ is the smallest enclosing disk of the cell $c$; and 
        \item\label{inv:4d} the barrier clearance disk $o(c_b)$ is disjoint from all disks $d'\in S_k$ associated with the cell~$c'$ with $c_b\prec c'$.
        \end{enumerate}
    \item\label{inv:5} If $d\in \mathcal{D}_k$ is associated with a cell $c_d\in N_k$ but $d\notin S_k$, then 
    \begin{enumerate}
        \item\label{inv:5a} $d$ intersects the obstacle disk $o(c')$ for some obstacle cell $c'$ with $c'\preceq c_d$, or
        \item\label{inv:5b} $d$ intersects a barrier clearance disk $o(c_b)$
        for some barrier~$b\in B_k$ with $c_b\preceq c_d$.
        \end{enumerate}
\end{enumerate}

We show (Lemma~\ref{lem:invariants} below) that invariants~\ref{inv:1}--\ref{inv:5} guarantee that the largest of the eight independent sets, $S_1,\ldots , S_8$, is a constant-factor approximate MIS of $\mathcal{D}$. As we use larger obstacle disks than in Section~\ref{sec:hierarchical}, to ensure 3-clearance, we need to adapt Lemma~\ref{lem:obstacle-packing} to the new setting. We prove the following with an easy volume argument.  

\begin{lemma}\label{lem:obstacle-packing2}
    Let~$c$ be a barrier or obstacle cell in a nonatree~$N_k$, $k\in \{1,\ldots, k\}$. Then the barrier clearance or obstacle disk $o(c)$ intersects at most $O(1)$ pairwise disjoint disks in higher buckets of $\mathcal{D}_k$.
\end{lemma}
\begin{proof}
 Let~$c$ be a barrier or obstacle cell at level $i\in \mathbb{Z}$. Then $c$ has side length $3^i$, and the disk associated with $c$ has radii in the range $(\frac14\, 3^{i-1},\frac14\, 3^i]$. 
 By invariant~\ref{inv:obsb}, the (barrier) obstacle disk of $c$ is $o(c)=3d'$, where $d'$ is the smallest enclosing disk of $c$. That is, the radius of $o(c)$ is $r=3\cdot 3^i/\sqrt{2}=3^{i+1}/\sqrt{2}$, and $\mathrm{area}(o(c))=\pi r^2=\frac{9\pi}{2} \, 3^{2i}$. 
 
 By Lemma~\ref{lem:sparse}, the radius of any disk $d\in \mathcal{D}_k$ in a higher bucket is at least $3\cdot \frac14\, 3^i=\frac14\, 3^{i+1}$. Let $\mathcal{A}$ be the set of pairwise disjoint disks in $\mathcal{D}_k$ in higher buckets that intersect $o(c)$. Scale down every disk $d\in \mathcal{D}$ from a point in $d\cap o(c)$ to a disk $\widehat{d}$ of radius $\widehat{r}=\frac14\, 3^{i+1}$; and let $\widehat{\mathcal{A}}$ be the set of resulting disks. Note that $|\widehat{\mathcal{A}}|=|\mathcal{A}|$, the disks in $\widehat{\mathcal{A}}$ are pairwise disjoint, and they all intersect $o(c)$. Let $D$ be a disk concentric with $o(c)$, of radius $R=3^{i+1}/\sqrt{2}+\frac12\, 3^{i+1}=\frac{1+\sqrt{2}}{2}\, 3^{i+1}$. 
 By the triangle inequality, $D$ contains all (scaled) disks in $\widehat{\mathcal{A}}$ (see Figure~\ref{fig:obstacle-packing}b for a congruent example).  
 Since the disks in $\widehat{\mathcal{A}}$ are pairwise disjoint, then $|\widehat{\mathcal{A}}|\cdot \pi\widehat{r}^2=\sum_{\widehat{d}\in\widehat{\mathcal{A}}} \mathrm{area}(\widehat{d}) \leq \mathrm{area}(D)=\pi R^2$. 
 This yields $|\mathcal{A}|=|\widehat{\mathcal{A}}| \leq \mathrm{area}(D)/\mathrm{area}(\widehat{d}) = R^2/\widehat{r}^2 = O(1)$, as claimed.
\end{proof}

We are now ready to prove that invariants~\ref{inv:1}--\ref{inv:5} ensure that one of the independent sets $S_1,\ldots , S_8$ is a constant-factor approximate MIS of $\mathcal{D}$.

\begin{lemma}\label{lem:invariants}
    Let $\mathcal{D}$ be a set of disks with the data structures described above, satisfying invariants~\ref{inv:1}--\ref{inv:5}. Then $\max_{1\leq k\leq 8} |S_k|\geq \Omega(|S^*|)$, where $S^*$ is a MIS of $\mathcal{D}$.  
\end{lemma}
\begin{proof} 
    Let $S^*\subset \mathcal{D}$ be a MIS, and let $S^*_k=S^*\cap \mathcal{D}_k$ for $k=1,\ldots, 8$. By invariant~\ref{inv:1}, we have $|S^*_k|\geq \frac18\, |S^*|$ for some $k\in \{1,\ldots ,8\}$. Fix this value of $k$ for the remainder of the proof. 

    Let $S^{**}_k\subset \mathcal{D}_k$ be a maximum  independent set subject to the constraints that (i) each grid cells in~$N_k$ contributes at most one disk to $S_k^{**}$, and (ii) any grid cell in $N_k$ that has two or more relevant children do not contribute. By Lemmata~\ref{lem:greedy} and~\ref{lem:14700+}, we have $|S^{**}_k|\geq \Omega(|S^*_k|)\geq \Omega(|S^*|)$.
    We claim that 
\begin{equation}\label{eq:claim}
    |S^{**}_k|\leq O(|S_k\cup B_k|),
\end{equation}
   This will complete the proof: Invariant~\ref{inv:4c} guarantees $|B_k|\leq 2\,|S_k|$. Since $S_k$ and $B_k$ are disjoint by invariant~\ref{inv:3a}, this yields $|S_k|\geq \frac13 (|S_k| + |B_k|) =\frac13\, |S_k\cup B_k|\geq \Omega(|S^{**}_k|)\geq \Omega(|S^*|)$, as required.

\smallskip\noindent\textbf{Charging Scheme.}
We prove \eqref{eq:claim}, using a charging scheme. 
Specifically, each disk $d^*\in S^{**}_k$ is worth one unit. We \emph{charge} every disk $d^*\in S^{**}_k$ to either a disk in $S_k\cup B_k$ or an obstacle cell, using Invariant~\ref{inv:5}. Note that the number of obstacle cells is at most $2\, |S_k|$ by Lemma~\ref{lem:skip-cells}. Then we show that the total number of charges received is $O(|S_k\cup B_k|)$, which implies $|S^*_k|\leq O(|S_k\cup B_k|)$, as required. 

In a bottom-up traversal, we consider the cells of the nonatree $N_k$. Consider each cell $c$ that contributes a disk to $S^{**}_k$, and consider the disk $d^*\in S^{**}_k$ associated with $c$. 
If there exists a disk $d'\in S_k\cup B_k$ associated with $c$, then we charge $d^*$ to such a disk $d'$.
Otherwise, Invariant~\ref{inv:5} provides two possible reasons why $d^*$ is not in $S_k$. We describe our charging scheme in each case separately:
\begin{enumerate}
\item[{\rm (a)}] If $d^*$ intersects the obstacle disk $o(c')$ for some cell $c'$ with $c'\preceq c$, then we first show that $c'\prec c$. Suppose, to the contrary, that $c=c'$. Since $c$ is not associated with any disk in $S_k \cup B_k$, but it is an obstacle cell, then $c$ has two or more relevant children, consequently no disk in $S_k^{**}$ is associated with $c$: A contradiction. 
We may now assume $c'\prec c$. By invariant~\ref{inv:obs}, there is a unique maximal obstacle disk $o(c')$ for some cell $c'$, $c'\prec c$, and we charge $d^*$ to the cell $c'$.  
\item[{\rm (b)}] Else $d^*$ intersects the barrier clearance disk $o(c_b)$, where cell~$c_b\preceq c$ is associated with some barrier disk $b\in B_k$. We charge $d^*$ to $b\in B_k$.
\end{enumerate}
We claim that each disk $d'\in S_k\cup B_k$ and each obstacle cell $c'$ receives $O(1)$ charges. The choice of the independent set $S_k^{**}$, each cell $c'$ contributes at most one disks to $S^{**}_k$.  
Consequently, at most one disk $d^*\in S^{**}_k$ is charged to $d'$. 

Consider now an obstacle cell $c'$. By Lemma~\ref{lem:obstacle-packing2}, an obstacle disk $o(c')$ intersects $O(1)$ pairwise disjoint disks of larger radii. Consequently, $O(1)$ disks $d^*\in S^{**}_k$ are charged to $c'$ using invariant~\ref{inv:5a}. Similarly, a barrier clearance disk $o(c_b)$, $b\in B_k$, intersects $O(1)$ pairwise disjoint disks of larger radii, and so at most $O(1)$ disks $d^*\in S^{**}_k$ are charged to any barrier disk $b\in B_k$ using invariant~\ref{inv:5b}. Overall, $|S^{**}_k|\leq |S_k\cup B_k|+O(|S_k|)+O(|B_k|)\leq O(|S_k\cup B_k|)$, as claimed.  
\end{proof}

Finally, we show a useful property of the obstacle disks. 

\begin{lemma}\label{lem:closest-obstacle}
    When a disk~$d$ in cell~$c\in N_k$ is added to $S_k$, it can intersect only the disk~$d_o\in S_k$ associated with the next obstacle cell~$c_o$ in the ascending path $P(d)$ from~$c$ towards the root, if $d_o$ even exists.
\end{lemma}
\begin{proof}
    Assume such a disk~$d$ is added to~$S_k$. Invariants~\ref{inv:obsa} and~\ref{inv:obsb} tell us, respectively, that $c$ must be an obstacle cell now, and the obstacle disk of~$c$ will be $3d'$, where $d'$ is the smallest disk enclosing~$c$. Now consider the next obstacle cell~$c_o$ along~$P(d)$, and observe that $c$ is completely contained in~$c_o$, and the same holds for their respective obstacle disks (see Figure~\ref{fig:closest-obstacle}b). Since all disks of~$S_k$ in even higher levels do not intersect the obstacle disk of~$c_o$, they cannot intersect the obstacle disk of~$c$ either. To conclude, consider the type of obstacle associated with~$c_o$: In case the obstacle is a merge obstacle, $d_o$ does not exist, while in case of a true obstacle, $d_o$ exists and $d$ may intersect~$d_o$.
\end{proof}

\subsection{Hierarchical Dynamic Cell Location and Farthest Neighbor Data Structures}
\label{sec:DDD}

For each nonatree $N_k$, $k=1,\ldots , 8$, we construct three point location---or rather cell location---data structures, $F_c$, $F_o$, and $F_o'$ and additionally a dynamic farthest neighbor data structures, $T_\cup$, described in this section. These data structures help navigate the nonatree: The cell location data structures~$F_c$ and~$F_o$ allow us to efficiently locate, respectively, a cell or an obstacle cell in the nonatree or the lowest ancestor if the query (obstacle) cell is non-existent; $F_o'$ returns for a given cell~$c$ the obstacle cell~$c_o$, $c_o\preceq c$ closest to~$c$. Furthermore, after deleting a disk associated with a cell $c_q$ from a current independent set $S_k$, the data structure $T_\cup$ helps find a cell $c$, $c_q\preceq c$, in which a new (disjoint) disk can be added to $S_k$. After adding a new disk $d_q$ associated with a cell $c_q$, we need to find the closest cell $c$, $c_q\prec c$, in which $d_q$ intersects some disk $d\in S_k$ (if any), which then has to be deleted from $S_k$. Lemma~\ref{lem:closest-obstacle} tells us this cell must be the closest ancestor obstacle cell, which can easily be found using~$F_o$.

Let $k\in \{1,\ldots, 8\}$ be fixed. Assume that $N_k$ is a nonatree, $\mathcal{D}_k$ is a set of disks associated with cells in $N_k$; and the sets $S_k$ and $B_k$ satisfy invariants~\ref{inv:1}--\ref{inv:5}. 

\begin{itemize}
\item The data structure $F_c$ is a point location data structure for~$N_k$ analogous to a $\mathcal{Q}$-order for compressed quad trees~\cite[Chapter~2]{har2011geometric}, which can be implemented in any ordered set data structure. We use $F_c$ only to locate cells, and hence we refer to it as a cell location data structure. The $\mathcal{Q}$-order corresponds to a depth-first search of~$N_k$, where sibling cells are ordered (geometrically) according to a Z-order (see Figure~\ref{fig:q-order}). When a cell~$c$ is not found, a pointer (\emph{finger}) to the closest ancestor is returned (an ancestor of~$c$, in case $c$ were to exist in~$N_k$). The data structure~$F_o$ functions exactly like~$F_c$ but consists of only obstacle cells.
\item The data structure~$F_o'$ mimics $F_o$, but whereas the DFS order of~$F_o$ corresponds to a pre-order tree walk of $N_k$, such that a parent cell comes before its children in the ordering, $F_o'$ corresponds to a post-order tree walk, and hence for a given cell~$c$, the closest obstacle cell~$c_o$ of~$N_k$, such that $c_o\preceq c$, is easily found.
\item The data structure $T_\cup$ is for all disks in $\mathcal{D}_k$. It supports insertions and deletions to/from $\mathcal{D}_k$; as well as the following query: Given a query cell $c_q$ of $N_k$ and an obstacle disk $o_q$, find the lowest cell $c$ such that $c_q\preceq c$ and there exists a disk $d\in \mathcal{D}_k$ associated with $c$ and disjoint from $o_q$, or report no such cell exists. The data structure $T_\cup$ is a hierarchical version of the DFN data structure (cf.~Lemma~\ref{lem:disjointness}).

\end{itemize}

\paragraph{Data Structures $F_c$, $F_o$, and $F_o'$.}
Since our (compressed) nonatree~$N_k$ is analogous to a compressed quadtree, the cell location data structure~$F_c$ works exactly like a $\mathcal{Q}$-order for compressed quadtrees: Insertion, deletion, and cell-location queries are therefore supported in $O(\log{n})$ time \cite[Chapter~2]{har2011geometric}. For completeness we show how to extend the quadtree $\mathcal{Q}$-order to nonatrees. 

The location of each cell of~$N_k$ is encoded in a binary number, and all cells are ordered according to their encoding. Let $L$ denote the list of levels of the nonatree $N_k$ in decreasing order. For the single cell on the top level of~$N_k$ we use an encoding of only zero bits, and on the second level we use the four most significant bits to encode the nine cells of~$N_k$, as shown in Figure~\ref{fig:q-order}a. Each subsequent level~$\ell\in L$ uses the next four bits to encode the 3x3 subdivision inside the cell encoded by the previous bits, as in Figure~\ref{fig:q-order}c. While~$F_c$ locates all uncompressed cells of~$N_k$, $F_o$ locates only the obstacle cells (which are uncompressed by definition).

\begin{figure}
    \centering
    \includegraphics{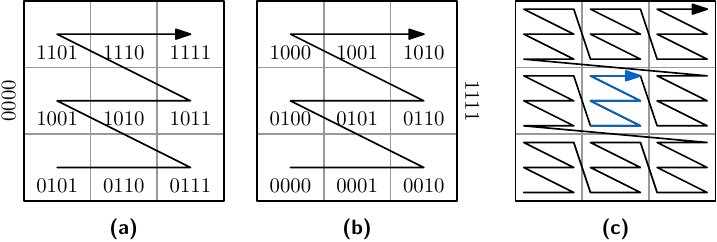}
    \caption{\textsf{$\mathcal{Q}$-orders in nonatrees. \textbf{(a)}} The encoding used in~$F_c$ and $F_o$, ordering the parent before its descendants. \textsf{\textbf{(b)}} The encoding used in~$F_o'$, ordering the parent after its descendants. \textsf{\textbf{(c)}} The recursive orders in~$F_c$: The blue arrow is encoded by $1010xxxx$; if $xxxx=0000$ then the middle cell~$c$ on the second level is located, and otherwise the four trailing bits determine the child of~$c$.}
    \label{fig:q-order}
\end{figure}

Finally, the data structure~$F_o'$ works exactly like~$F_o$, but uses a slightly different encoding that allows a parent cell to be ordered after its ancestors, instead of before, as shown in Figure~\ref{fig:q-order}b. This property is crucial for our usage of~$F_o'$: We will query~$F_o'$ with a cell~$c$ to find the closest obstacle cell~$c_o$, $c_o\preceq c$, whose obstacle helps us determine which disks in higher levels of~$N_k$ (at least as high as~$c$) can be added to~$S_k$ without overlapping (the 3-clearance of) disks in~$S_k$ in the levels below~$c$. The returned obstacle cell is always uniquely defined because we maintain invariant~\ref{inv:obsa}. Either~$c$ has multiple subtrees in which an obstacle is defined, in which case $c$ must be an obstacle cell itself, or the closest obstacle cell~$c_o$ below (and including)~$c$ is located in the single subtree rooted at~$c$ containing all relevant cells contained in~$c$. In both cases, a cell location query with~$c$ will either find~$c_o= c$ as the obstacle cell we are looking for, or returns the predecessor of~$c$, which is the closest obstacle cell~$c_o$ in the single subtree rooted at~$c$. 

\paragraph{Data Structure $T_\cup$.}
Let $L$ denote the list of levels of the nonatree $N_k$ in increasing order. The \emph{weight} $w(\ell)$ of a level $\ell\in L$ is the number of disks in $\mathcal{D}_k$ associated with cells in level $\ell$. In particular, the sum of weights is  $\sum_{\ell\in L} w(\ell)=|\mathcal{D}_k|$.

\begin{figure}
    \centering
    \includegraphics{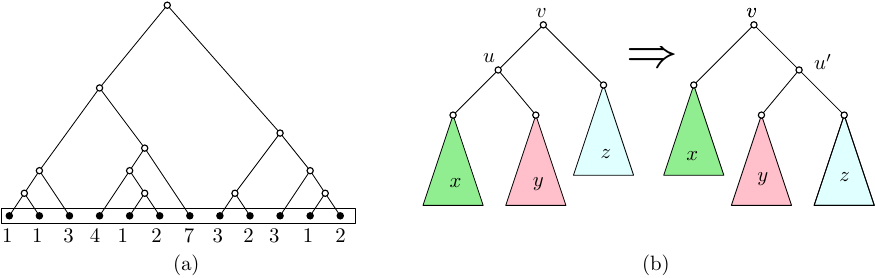}
    \caption{\textsf{\textbf{(a)}} An example for a weight-balanced binary tree over the levels $L$ of the nonatree $N_k$; and the weights the leaves leaves. 
    \textbf{\textsf{(b)}} A right rotation at node $v$ of $T_\cup$.}
    \label{fig:Tcup}
\end{figure}

Let $T_\cup(v)$ be a \emph{weight-balanced binary search tree} on $L$~\cite[Sec.~3.2]{Brass2008}; see Figure~\ref{fig:Tcup}a. That is, $T_\cup(v)$ is a rooted tree, where the leaves correspond to the elements of $L$, and each internal node corresponds to a sequence of consecutive leaves in $L$. The \emph{weight} of a subtree $T_\cup(v)$ rooted at a node $v$, denoted $w(T_\cup(v))$, is the sum of the weights of the leaves in $T_\cup(v)$. The weight-balance is specified by a parameter $\alpha\approx 0.29$, as follows: For each subtree, the left and right sub-subtrees each have at least $\alpha$ fraction of the total weight of the subtree, or is a singleton (i.e., a leaf) of arbitrary weight. It is known that a weight-balanced tree with total weight $n$ has height $O(\log n)$, and supports \emph{insertion} and \emph{deletion} of leaves using $O(\log n)$ rotations (Figure~\ref{fig:Tcup}b). Furthermore, the time from one rotation of a node $v$ to the next rotation of $v$, a positive fraction of all leaves below $v$ are changed~\cite[Sec.~3.2]{Brass2008}.

Recall that each node $v$ of $T_\cup$ corresponds to a sequence of consecutive levels of the nonatree $N_k$. Let $\mathcal{D}_k(v)$ denote the set of all disks of $\mathcal{D}_k$ on these levels; and let $G(v)$ denote the grid corresponding to the highest of these levels. For each cell $c\in G(v)$, let $\mathcal{D}(v,c)$ denote the set of disks in $D_k(v)$ that lie in the cell $c$. Now node $v$ of the tree $T_\cup$ stores, for each nonempty cell $c\in G(v)$, the DFN data structure for $\mathcal{D}_k(v,c)$. 

\begin{lemma}\label{lem:Tcup}
    The data structure $T_\cup$ supports insertions and deletions of disks in $\mathcal{D}_k$ in $O(\log^{10} n)$ expected amortized time, as well as the following query in $O(\log^3 n)$ worst-case time: Given a query cell $c_q$ of $N_k$ and an obstacle disk $o_q$, find the lowest cell $c$ such that $c_q\preceq c$ and there exists a disk $d\in \mathcal{D}_k$ associated with $c$ and disjoint from $d_q$, or report that no such cell exists.  
\end{lemma}
\begin{proof}
As noted above, for each internal node $v$, the weight of the two subtrees must change by $\Omega(w(T_\cup(v)))$ between two consecutive rotations. For a rotation at node $v$, we recompute all DFN data structure at the new child of $v$. Specifically, consider w.l.o.g.\ a right rotation at $v$ (refer Figure~\ref{fig:Tcup}b), where the left child $u$ is removed and the new right child $u'$ is created. The DFN data structures at $v$, $x$, $y$, and $z$ remain valid, and we need to compute a new DFN data structure for $u'$. By Lemma~\ref{lem:disjointness}, the expected preprocessing time of the DFN data structure for a set of $m$ disks is $O(m\log^9 m)$. This means that the update of the data structure $T_\cup$ due to a rotation at a node $v$ of weight $m=w(u')\leq O(n)$ takes $O(m\log^9 m)$ 
time. Consequently, rotations at a node $v$ of weight $w(v)$ can be done in $O(\log^9 w(v))\leq O(\log^9 n)$ expected amortized time. 

Thus, a rotation on each level of $T_\cup$ take $O(\log^9 n)$ expected amortized time per insertion or deletion. Summation over $O(\log n)$ levels implies that $O(\log^{10} n)$ expected amortized update time is devoted to rotations. Note also that for a disk insertion or deletion, we also update $O(\log n)$ DFN data structures (one on each level of $T_\cup$), each of which takes $O(\log^9 n)$ expected amortized update time. Overall, the data structure $T_\cup$ supports disk insertion and deletion in $O(\log^{10} n)$ amortized expected time.  

For a query cell $c_q$ and disk $d_q$, consider the ascending path in $T_\cup$ from the level of $c_q$ to the root. Consider the right siblings (if any) of all the nodes in this path. For each right sibling $v$, there is a unique cell $c_v\in G(v)$ such that $c_q\subseteq c_v$. We query the DFN data structure for $\mathcal{D}_k(v,c_v)$. If none of these DFN data structures finds any disk in $\mathcal{D}_k(v,c_v)$ disjoint $d_q$, then report that all disks associated with the ancestor cells of $c_q$ intersect $d_q$. Otherwise, let $v$ be the first (i.e., lowest) right sibling in which the DFN data structure returns a disk $d_v\in \mathcal{D}_k(v,c_v)$ disjoint $d_q$. By a binary search in the subtree $T_\cup(v)$, we find a leaf node $\ell\in L$ in which 
the DFN data structure returns a disk $d_\ell\in \mathcal{D}_k(\ell,c_\ell)$ disjoint $d_q$. In this case, we return the cell $c_\ell$ and the disk $d_\ell$. 
By Lemma~\ref{lem:disjointness}, we answer the query correctly, based on $O(\log n)$ queries to the DFN data structures, which takes $O(\log n)\cdot O(\log^2 n)=O(\log^3 n)$ worst-case time.
\end{proof}

\subsection{Dynamic Maintenance Using Dynamic Farthest Neighbor Data Structures}
\label{sec:dynamic}
To maintain an approximate maximum independent set of disks, we now consider how our data structures are affected by updates: Disks are inserted and deleted into an initially empty set of disks, and our goal is to 
maintain the data structures described in Section~\ref{sec:dynamization} and Section~\ref{sec:DDD}. 

On a high level, for a dynamic set of disks $\mathcal{D}$, we maintain eight nonatrees $N_1,\ldots N_8$, and for each $k\in \{1,\ldots , 8\}$, we maintain the cell location data structures~$F_c$, $F_o$, and~$F_o'$ and two sets of disks: an independent set~$S_k$ and a set of barrier disks~$B_k$. In this section, we show how to maintain these data structures with polylogarithmic update times while maintaining invariants~\ref{inv:1}--\ref{inv:5} described in Section~\ref{sec:dynamization}. For that, we may use the additional data structure~$T_\cup$, as defined in Section~\ref{sec:DDD}, to efficiently query the nonatrees, and their independent sets and barrier disks.


For maintaining the invariants of our solution, 
we can deal with each ascending path independently: If a disk in cell~$c$ on an ascending path is added to~$S_k$, then we create a new (true) obstacle cell with obstacle~$o(c) = 3d'$ where $d'$ is the smallest enclosing disk of~$c$. Observe that this disk is a subset of the obstacle disk at the top end of the ascending path, since the obstacle cell at the top has strictly larger side length. We ensure that every disk in~$S_k$ does not intersect an obstacle disk below it in the nonatree, and hence if a disk above the ascending path would intersect $o(c)$, then it would also intersect the obstacle of the obstacle cell at the top of the ascending path (cf.~Lemma~\ref{lem:closest-obstacle}). Thus when adding disks to~$S_k$, changes are contained within an ascending path. 

More specifically, when a disk $d$ associated with a cell $c\in N_k$ is inserted or deleted, then $c$ lies in an ascending path~$P(d)$ between two obstacle cells, say $c_1\preceq c\prec c_2$. To update the independent set $S_k$ and the barrier disks $B_k$, in general we run the greedy algorithm in this path. The greedy process guarantees that these disks are disjoint from any smaller disk in $S_k$. However, the newly added disks in $S_k$ may intersect the disk $s_2\in S_k$ associated with $c_2$: If this is the case, we delete $s_2$ from $S_k$, insert it into $B_k$, and assign it to the highest disk in~$S_k$ in $P(d)$ below $s_2$; this highest disk in $P(d)$ is necessarily the disk added last to~$S_k$, causing the intersection with~$s_2$.
Note, however, that if $s_2$ was already associated with a barrier disk, $\beta(c_2)$, then adding $s_2$ to $B_k$ would violate invariant~\ref{inv:4c}. For this reason, if $\beta(c_2)$ exists, we remove $s_2$ from $S_k$, run the greedy algorithm on a longer path, up to the cell associated with $\beta(c_2)$, and then reassign $\beta(c_2)$ to the largest disk in $S_k$ found by the greedy algorithm. Overall we distinguish between three cases to handle these scenarios (cf.~step~\ref{upd:while-loop} of \texttt{UpdateIS} below). 

We are now ready to explain for each insertion/deletion of a disk into~$\mathcal{D}$, how we update our data structures, beginning with a few useful subroutines:

\paragraph{Adding to and Removing from~$S_k$.} We start by specifying two subroutines that we use often to update an independent set~$S_k$. When we add a disk to or remove a disk from~$S_k$, we need to interact with the data structures~$T^i_k$, $F_o$, and~$F_o'$. These subroutines help abstracting from those data structures in the upcoming routines. First we explain how to \texttt{Add} a disk~$d$ in cell~$c$ to~$S_k$:

\begin{enumerate}
    \item Insert~$d$ into~$T^i_k$.
    \item\label{add:obs} Create the obstacle~$o(c)$ and insert~$c$ into~$F_o$ and~$F_o'$ (if it is not in these yet).
\end{enumerate}

Second, we explain the subroutine to \texttt{Remove} a disk~$d$ in cell~$c$ from $S_k$:

\begin{enumerate}
    \item Delete~$d$ from~$T^i_k$.
    \item\label{remove:obs} Remove the obstacle~$o(c)$ and delete~$c$ from~$F_o$ and~$F_o'$.
\end{enumerate}

\paragraph{Greedy Independent Set Procedure.} The subroutine \texttt{GreedyIS} runs the greedy algorithm on an ascending path in the nonatree~$N_k$ between the cells $c_1$ and $c_2$, $c_1 \prec c_2$, finds pairwise independent disks that are also disjoint from the obstacle disk $o(c_1)$, and returns the highest obstacle in the path at the end of the process. 

\begin{enumerate}
    \item\label{upd:find-cell} Query~$T_\cup$ with $c_1$ and $o(c_1)$ to either find the lowest cell~$c^*$ (in bucket $i$ of~$N_k$), such that $c_1\preceq c^*$, together with a disk $d\in\mathcal{D}_k$ associated with $c^*$ and disjoint from~$o(c_1)$, or find that no such cell and disk exist.
    \item\label{gr:while-loop} While $c^*$ exists and querying~$F_o'$ with $c^*$ returns a cell~$c \prec c_2$, repeat the following steps:
    \begin{enumerate}
        \item\label{upd:ins-Sk} First \texttt{Add} disk~$d$ to~$S_k$.
        \item\label{upd:set-obs} Rename $c^*$ to $c_o$.
        \item\label{upd:find-cell2} Query~$T_\cup$ with $c_o$ and $o(c_o)$ to either find the lowest cell~$c^*$ (in bucket $i$ of~$N_k$), such that $c_o\preceq c^*$, together with a disk $d\in\mathcal{D}_k$ associated with $c^*$ and disjoint from~$o(c_o)$, or find that no such cell and disk exist. 
    \end{enumerate}
    \item Return $o(c_o)$.
\end{enumerate}

\paragraph{Updating Independent Sets.} Subroutine \texttt{UpdateIS} finds, for a cell~$c$ in $N_k$, the ascending path in the nonatree that contains $c$, and then runs \texttt{GreedyIS}, distinguishing between three cases based on whether the obstacle cell at the top of the path is associated with a disk in $S_k$ and whether it assigned a barrier. 
Several steps of \texttt{UpdateIS} are visualized in Figure~\ref{fig:updateIS}.
\begin{enumerate}
    \item\label{upd:find-obs} Query~$F_o'$ with $c$ to find the highest obstacle cell~$c_o$, with $c_o \preceq c$.
    \item\label{upd:empty-Bk} If $\beta(c_o)\neq\textsc{nil}$, then remove $\beta(c_o)$ from~$B_k$ and set $\beta(c_o):=\textsc{nil}$.
     
    \item\label{upd:while-loop} Query~$F_o$ with the parent of~$c^*$, to find the lowest obstacle cell~$c^-$ such that $c^*\prec c^-$
    \begin{enumerate}
        \item\label{upd:case1} If no disk~$d^-\in S_k$ is associated with~$c^-$, then call \texttt{GreedyIS} with $c_1=c_o$ and $c_2=c^-$.  
        \item\label{upd:case2} Else if a disk~$d^-\in S_k$ is associated with~$c^-$, but $\beta(c^-)$ does not exist, then call \texttt{GreedyIS} with $c_1=c_o$ and $c_2=c^-$, which returns an obstacle disk $o(c)$ for some cell $c_1\prec c\prec c_2$. If $o(c)$ intersects $d^-$, then 
        \texttt{Remove} $d^-$ from $S_k$, add it to $B_k$, and set $\beta(c)=d^-$.
        \item\label{upd:case3} Else (a disk~$d^-\in S_k$ is associated with~$c^-$, and $\beta(c^-)$ exists), then \texttt{Remove} $d^-$ from $S_k$.         
        Call \texttt{GreedyIS} with $c_1=c_o$ and $c_2$ being the cell associated with $\beta(c^-)\in B_k$, which returns an obstacle disk $o(c)$ for some $c_1\prec c\prec c_2$. 
        If $o(c)$ intersects $\beta(c^-)$, then set $\beta(c)=\beta(c^-)$, else \texttt{Add} $\beta(c^-)$ to $S_k$ and remove $\beta(c^-)$ from $B_k$. 
        In both cases set $\beta(c^-)=\textsc{nil}$.
    \end{enumerate}
\end{enumerate}

\begin{figure}
    \centering
    \includegraphics[page=2]{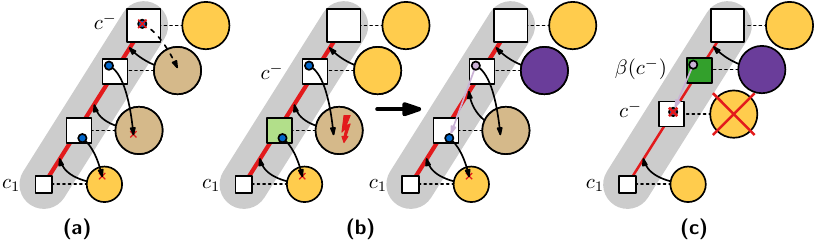}
    \caption{The three cases in Step~\ref{upd:while-loop} of \texttt{UpdateIS} ensure that no two barriers exist between consecutive obstacle cells in the gray ascending path(s): \textbf{\textsf{(a)}} There is no disk~$d^-\in S_k$ in~$c^-$ that can intersect the new (brown) obstacle disks in the gray ascending path. \textbf{\textsf{(b)}} The disk $d^-\in S_k$ in~$c^-$ is turned into a barrier if it overlaps the obstacle disk of the highest new disk in the light green cell. \textbf{\textsf{(c)}} If~$\beta(c^-)$ exists, remove~$d^-$ from~$S_k$ and run \texttt{GreedyIS} up to the dark green cell.}
    \label{fig:updateIS}
\end{figure}

\paragraph{Insertion.} 

Let $d_q$ be a disk that is inserted in step~$q$, we make the following updates:
\vspace{.2cm}

We \texttt{Insert} $d_q$ into our data structures, which relies on subroutine \texttt{UpdateIS}. 
\begin{enumerate}
    \item\label{ins:find-cell} Find the bucket~$i$ that $d_q$ belongs to based on its radius, and find a unique cell~$c$ in bucket~$i$ that fully contains $d_q$, using the center point of~$d_q$, similar to the unit disk case (see Figure~\ref{fig:grids}).
    \item\label{ins:find-nonatree} Determine the nonatree~$N_k$ that~$d_q$ will be inserted into: The bucket~$i$ determines whether we insert into a nonatree consisting of odd or even buckets, and the cell~$c$ determines which shifted grid, and hence which of the four nonatrees of the appropriate parity we insert into. 
    \item\label{ins:ins-nonatree} Add~$d_q$ to~$N_k$: Remember that $N_k$ is compressed, and hence we need to first locate~$c$ in $N_k$ (in $O(\log{n})$ time). If $c$ does not exist, then the cell location query with~$c$ finds the lowest ancestor~$c_a$ of~$c$. Analogous to compressed quadtrees, inserting~$c$ as a descendant of~$c_a$ requires at most a constant number of other cells to be updated; these are either split or updated in terms of parent-child relations. We now make the following changes:
    \begin{enumerate}
        \item\label{ins:obs-leaf} If~$c$ is a leaf that is not in~$F_o$ and~$F_o'$ yet, \texttt{Add} an arbitrary disk $d$ associated with~$c$ to $S_k$. Add $c$ to~$F_o$ and~$F_o'$. 
        \item\label{ins:obs-misc} For any cell~$c'$ of the $O(1)$ cells that may get new children by updated parent-child relationships, we query~$F_o'$ with the child cells, to find out whether $c'$ has two or more relevant children.
        If so, create the obstacle disk $o(c')$ and insert $c'$ into $F_o$ and~$F_o'$ (if it was not in these yet), and call subroutine \texttt{UpdateIS} on~$c'$.
        \item\label{ins:delete_Bk1} Furthermore, before creating $o(c')$ and calling subroutine \texttt{UpdateIS} on~$c'$, we query~$T^i_k$ with~$c'$ to find whether it is associated with a disk~$d\in S_k$, if so, \texttt{Remove} $d$ from~$S_k$. 
        \item\label{ins:delete_Bk2} If $c'$ had only a single subtree with relevant children before, let $c_o$ be the obstacle cell found by~$F_o'$ by querying with the root of this subtree. 
        If $\beta(c_o)\neq\textsc{nil}$ and $c'$ lies between $c_o$ and the cell associated with $\beta(c_0)$, then remove $\beta(c_o)$ from $B_k$, and set $\beta(c_o):=\textsc{nil}$.
    \end{enumerate}
    To finalize this step, insert~$d_q$ into~$T^i_\mathcal{D}$, such that we can find~$d_q$ after querying~$T^i_\mathcal{D}$ for cell~$c$. Similar to the unit disk case, this may require us to insert a node for cell~$c$ in ~$T^i_\mathcal{D}$ first.
    \item\label{ins:ins-Tcup} Insert~$d_q$ into~$T_\cup$ of~$N_k$. Precisely, $d_q$ is inserted into the DFN data structure of cell~$c$ in the leaf~$t$ of~$T_\cup$ corresponding to bucket~$i$. Note that if there was no leaf for bucket~$i$ yet, then this node is created and $T_\cup$ may be rebalanced (all in $O(\log^9{n})$ expected amortized time). Additionally, if $c$ did not exist yet, then the DFN data structure is initialized in this step. Subsequently, $d_q$ is added to all $O(\log{n})$ nodes on the path from~$t$ to the root of~$T_\cup$. In particular, $d_q$ is inserted in all DFN data structures of these nodes corresponding to the cell (of a coarser grid) that overlaps~$c$.   
    \item\label{ins:call-update} Call subroutine \texttt{UpdateIS} on cell~$c$.
\end{enumerate}

\paragraph{Deletion.} Let $d_q$ be a disk that is deleted in step~$q$, we make the following updates:
\vspace{.2cm}

We \texttt{Delete} $d_q$ from our data structures, which again relies on the subroutine \texttt{UpdateIS}.
\begin{enumerate}
    \item\label{del:find-cell} Find the bucket~$i$ that $d_q$ belongs to based on its radius, and find a unique cell~$c$ in bucket~$i$ that fully contains $d_q$, using the center point of~$d_q$, similar to the unit disk case (see Figure~\ref{fig:grids}).
    \item\label{del:find-nonatree} Determine the nonatree~$N_k$ that~$d_q$ is located in: The bucket~$i$ determines whether we insert into a nonatree consisting of odd or even buckets, and the cell~$c$ determines which shifted grid, and hence which of the four nonatrees of the appropriate parity we insert into. 
    \item\label{del:del-nonatree} Remove~$d_q$ from~$N_k$: We first locate~$c$ in $N_k$.
    If $c$ does not exist, we are done (since $d_q$ does then not exist in~$N_k$). If $c$ exists, let~$c_p$ be the parent cell of~$c$ (if such a parent exists).
    We delete~$d_q$ from~$T^i_\mathcal{D}$, and query it with~$c$ to check if the cell is now empty. If so, we delete~$c$ from~$T^i_\mathcal{D}$ and also from~$N_k$, which requires at most a constant number of other cells in~$N_k$ to be updated; these are either merges or updated in terms of parent-child relations. Note that these merges can merge $c_p$, in case it is empty, with other (empty) siblings, and we consider~$c_p$ to be the lowest existing non-empty ancestor of~$c$. We now make the following changes:
    \begin{enumerate}
        \item\label{del:obs-leaf} If~$c_p$ is a leaf, query~$T^i_k$ with~$c_p$ to find whether a disk~$d_p\in S_k$ is associated with~$c_p$, if not, \texttt{Add} an arbitrary disk associated with~$c_p$ to $S_k$. 
        \item\label{del:obs-unmerge}
        If~$c_p$ is a merge obstacle cell,   query~$F_o'$ with the child cells of~$c_p$,    to find out whether there are obstacles (and hence relevant cells) in at least two subtrees. If not, also remove the obstacle~$o(c)$ and delete~$c$ from~$F_o$ and~$F_o'$, remove~$\beta(c)$ from $B_k$, and set $\beta(c):=\textsc{nil}$. 
    \end{enumerate}
    \item\label{del:del-Tcup} Delete~$d_q$ from~$T_\cup$ of~$N_k$. Precisely, $d_q$ is deleted from the DFN data structure of cell~$c$ in the leaf~$t$ of~$T_\cup$ corresponding to bucket~$i$. Note that if $c$ is now empty, then the DFN data structure for it can be removed. Additionally, if the leaf for bucket~$i$ is now empty, then this node is deleted and $T_\cup$ may be rebalanced (in $O(\log^9{n})$ expected amortized time). Subsequently, $d_q$ is removed from all $O(\log{n})$ nodes on the path from~$t$ to the root of~$T_\cup$. In particular, $d_q$ is deleted from all DFN data structures of these nodes corresponding to the cell (of a coarser grid) that overlaps~$c$, again removing DFN data structures when empty.   
    \item\label{del:del-Sk} Query~$T^i_k$ with $c$ to find whether~$d_q\in S_k$. If so, \texttt{Remove}~$d_q$ from~$S_k$, remove~$\beta(c)$ from $B_k$, and set $\beta(c):=\textsc{nil}$. 
    \item\label{del:del-Bk} Query $F_o'$ with $c$ to find the highest obstacle cell $c_o$, with $c_o \preceq c$. If $\beta(c_o)=d_q$, remove~$d_q$ from~$B_k$, and set $\beta(c):=\textsc{nil}$. 
    \item\label{del:call-update} Call subroutine \texttt{UpdateIS} on cell~$c$, if it is not deleted. Otherwise call \texttt{UpdateIS} on cell~$c_p$.
\end{enumerate}

\paragraph{Update Time Analysis.} Subroutine \texttt{UpdateIS} runs a greedy algorithm on an ascending path $P$ of $N_k$ between two consecutive obstacle cells. 
We show first that the number of iterations in the while loop of \texttt{GreedyIS} (step~\ref{gr:while-loop}) is bounded by  $O(1)$.
\begin{lemma}\label{lem:LocalGreedy}
In each call to \texttt{GreedyIS}, the while loop  terminates after $O(1)$ iterations. 
\end{lemma}
\begin{proof}
Assume that \texttt{GreedyIS} is called for an ascending path $P$ between obstacle cells $c_1\prec c_2$.  Let $S^*$ be the set of disks added to $S_k$ in the while loop of \texttt{GreedyIS}. Then the while loop has $|S^*|+1$ iterations. We need to show that $|S^*|\leq O(1)$. 

By invariant~\ref{inv:5}, all disks associated with the cells of the ascending path $P$ between $c_1$ and $c_2$ intersect an obstacle disk or an barrier clearance disk $o(c)$ for some cell $c_1\preceq c\prec c_2$. We call such (barrier) obstacle disks the \emph{dominating} disks of~$P$. By Lemma~\ref{lem:obstacle-packing2} each dominating disk intersects $O(1)$ pairwise disjoint disks. Thus it is enough to show that $P$ has $O(1)$ such dominating disks.

Before any updates are made, any ascending path $P_0$ between two consecutive obstacle cells contains at most two dominating disks of $P_0$, the obstacle disk of the bottom cell of the path, and at most one barrier clearance disk, by invariant~\ref{inv:obs}. The ascending path $P$ is contained in the union of $O(1)$ such paths $P_0$.

Specifically, \texttt{GreedyIS} is called only in step~\ref{upd:while-loop} of \texttt{UpdateIS}, in three possible cases. In cases~\ref{upd:case1} and \ref{upd:case2}, $P$ is a path between two obstacle cells, and in case~\ref{upd:case3} it is contained in the union of two such paths. However, \texttt{UpdateIS} is called after the insertion or deletion of up to one obstacle cell in \texttt{Insert} and \texttt{Delete}. Consequently, $P$ is contained in the union of up to four ascending paths of the initial nonatree (i.e., the nonatree before the current update), and so $P$ contains at most $O(1)$ dominating disks, as required. 
%
\end{proof}

\begin{lemma}\label{lem:greedyIS}
   Subroutine \texttt{GreedyIS} takes polylogarithmic amortized expected update time.
\end{lemma}
\begin{proof}
The while loop in step~\ref{gr:while-loop} is repeated $O(1)$ times by Lemma~\ref{lem:LocalGreedy}. 
It queries the data structure $T_\cup$ in step~\ref{upd:find-cell} and in each iteration of step~\ref{upd:find-cell2}. Each query, and $O(1)$ queries jointly, take polylogarithmic expected amortized time by Lemma~\ref{lem:Tcup}. 
Each iteration of the while loop queries $F_o'$ in $O(\log n)$ time and calls \texttt{Add}.
Subroutine \texttt{Add} interacts with the data structure~$F_o'$, $F_o$, and $T_k^i$ in polylogarithmic time. Overall,  \texttt{GreedyIS} takes polylogarithmic amortized expected update time.
\end{proof}

\begin{lemma}\label{cor:update}
   Subroutine \texttt{UpdateIS} takes polylogarithmic amortized expected update time.
\end{lemma}
\begin{proof}
Steps~\ref{upd:find-obs} and~\ref{upd:while-loop} of \texttt{UpdateIS} query the data structure~$F_o'$ and $F_o$ in in $O(\log{n})$ time each. 
Step~\ref{upd:empty-Bk} manipulates barrier disks in $O(1)$ time. 
Each of the three cases in step~\ref{upd:while-loop} calls \texttt{GreedyIS}, and it may also run \texttt{Add} or \texttt{Remove} once. 
Subroutines \texttt{Add} and \texttt{Remove} each interact with the data structure~$F_o'$, $F_o$, and $T_k^i$ in polylogarithmic time. 
Subroutine \texttt{GreedyIS} takes polylogarithmic amortized expected update time by Lemma~\ref{lem:greedyIS}. Overall, \texttt{UpdateIS} runs a polylogarithmic expected amortized time. 
\end{proof}

\begin{lemma}\label{lem:insert-time}
Dynamic insertion of a disk takes polylogarithmic amortized expected update time; 
and incurs $O(1)$ insertions to and deletions from $S_k$ for some  $k\in \{1,\ldots, 8\}$. 
\end{lemma}
\begin{proof}
Steps~\ref{ins:find-cell} and~\ref{ins:find-nonatree} of the \texttt{Insert} routine take $O(1)$ time. Step~\ref{ins:ins-nonatree} takes $O(\log^2{n})$ time, since cell location using~$F_c$, and insertion into~$N_k$ are handled in logarithmic time, while insertion into search tree~$T^i_\mathcal{D}$ takes $O(\log{n})$ time. Similarly, the $O(1)$ interactions with $F_o$ and~$F_o'$ in steps~\ref{ins:obs-leaf} and~\ref{ins:obs-misc} take logarithmic time as well. Step~\ref{ins:ins-Tcup} takes polylogarithmic expected amortized time by Lemma~\ref{lem:Tcup}.
Finally, the subroutine \texttt{UpdateIS} is called up to twice, in steps~\ref{ins:obs-misc} and~\ref{ins:call-update}, which also runs in polylogarithmic amortized expected update time by Lemma~\ref{cor:update}. 
\end{proof}

\begin{lemma}\label{lem:delete-time}
Dynamic deletion of a disk takes polylogarithmic expected amortized update time; 
and incurs $O(1)$ insertions to and deletions from $S_k$ for some $k\in \{1,\ldots, 8\}$. 
\end{lemma}
\begin{proof}
Steps~\ref{del:find-cell}-\ref{del:del-Tcup} of \texttt{Delete} have the same asymptotic running time as the corresponding steps in \texttt{Insert}, since (asymptotically) the same number of interactions take place with the same data structures. Steps~\ref{del:del-Sk} and~\ref{del:del-Bk} query~$T^i_k$ and $F_o'$, resp., possibly call subroutine \texttt{Remove}, in $O(\log{n})$ time each, and manipulate barrier disks in $O(1)$ time. Finally, the subroutine \texttt{UpdateIS} is called in step~\ref{del:call-update}, which runs in polylogarithmic amortized expected update time by Lemma~\ref{cor:update}. 
\end{proof}

\paragraph{Maintenance of Invariants.} We show that invariants~\ref{inv:1}--\ref{inv:5} hold after each \texttt{Insert} or \texttt{Delete} call.

\begin{lemma}\label{lem:maintain-invariants}
    Dynamic insertion or deletion of a disk maintains invariants~\ref{inv:1}--\ref{inv:5}.
\end{lemma}
\begin{proof}
    For this proof we assume that the invariants hold before a dynamic update, and show that after an insertion or deletion, the invariants still hold.
    
    \paragraph{Invariants~\ref{inv:1} and~\ref{inv:2}.}Steps~\ref{ins:find-cell} and~\ref{ins:find-nonatree} of both \texttt{Insert} and \texttt{Delete} ensure that invariants~\ref{inv:1} and~\ref{inv:2} are satisfied, by identifying exactly one cell~$c$ in one nonatree~$N_k$ (with the right level parity), for the inserted/deleted disk. 
    
    \paragraph{Invariant~\ref{inv:3}.} 
    For invariant~\ref{inv:3a}, we need to show that no disk lies in both $S_k$ and $B_k$. New disks are added to~$S_k$ in four steps of the algorithm: in step~\ref{upd:ins-Sk} of \texttt{GreedyIS}, in step~\ref{ins:obs-leaf} of \texttt{Insert} and step~\ref{del:obs-leaf} of \texttt{Delete} at a newly created leaf of the nonatree, and in step~\ref{upd:case3} of \texttt{UpdateIS} where a disk moves from $B_k$ to $S_k$. The latter three steps clearly keep $S_k$ and $B_k$ disjoint. Subroutine \texttt{GreedyIS} is called only in \texttt{UpdateIS}. Before we call \texttt{GreedyIS} on an ascending path, we always remove any barrier disks from that path in step~\ref{upd:empty-Bk}; as well as already before calling \texttt{UpdateIS}: in steps~\ref{ins:delete_Bk1} and~\ref{ins:delete_Bk2} of \texttt{Insert}; and steps~\ref{del:del-Sk} and~\ref{del:del-Bk} of \texttt{Delete}. 
    This establishes that $B_k$ and $S_k$ remain disjoint (invariant~\ref{inv:3a}). 

    For invariant~\ref{inv:3b}, note that when we add a new disk $d$ to $S_k$, the cells associate with $d$ contribute disks to neither $S_k$ nor $B_k$. A new disk is added to~$B_k$ only in step~\ref{upd:case2} of \texttt{UpdateIS}. In this step, a disk $d^-$ is moved from $S_k$ to $B_k$. Consequently, the cell associated with $d^-$ does not contribute any other barrier disk to $B_k\cup S_k$, and so invariant~\ref{inv:3b} is maintained.     

    It remains to consider invariant~\ref{inv:3c}, that $S_k$ is an independent set with 3-clearance. Assume that this holds before the dynamic update; and note that deleting disks from $S_k$ cannot violate this property. We therefore have to prove only for a newly added disk~$d$ that the invariant still holds. 

\begin{itemize}
    \item Suppose a new disk $d$ is added to $S_k$ in step~\ref{upd:ins-Sk} of \texttt{GreedyIS}. 
    We query~$T_\cup$ with a cell~$c$ and an obstacle~$o(c_o)$ to find~$d$ in cell~$c^*$ (in step~\ref{upd:find-cell} or~\ref{upd:find-cell2} of \texttt{GreedyIS}). The obstacle~$o(c_o)$ was found by querying~$F_o'$ with~$c$, and hence~$c_o$ is the closest obstacle cell below~$c$. Any obstacle cell below~$c_o$ would have an obstacle with smaller radius (being empty, or defined by a cell~$c' \prec c_o$), which would also be completely contained in~$o(c_o)$ (because $c' \prec c_o$). Thus, disk~$d$ found by~$T_\cup$ must be disjoint from any obstacle cell below~$c$, satisfying invariant~\ref{inv:3c} for the disk of~$S_k$ in those obstacle cells below~$c$.

    For the disks above~$c$, we know the following. In step~\ref{upd:find-cell} of \texttt{GreedyIS}, we check whether the closest obstacle disk below the cell~$c^*$ is still~$o(c_o)$. If this is not the case, then a disk is found above the next obstacle disk on the ascending path from~$c$ to the root. Adding such a disk is problematic, since it will intersect with the obstacle disk above $c$. Thus we add~$d$ only if it is located between $c_o$ and the next obstacle cell above it. All other disks in $S_k$ on the ascending path from~$c$ to the root are therefore located above $d$. By Lemma~\ref{lem:clearance}, at most one disk in~$S_k$ can intersect the obstacle of cell~$c^*$. If such a disk exists, it must be in the cell returned by~$F_o$ (Lemma~\ref{lem:closest-obstacle}). However, if $o(c^*)$ intersects a disk $d^-$ in the next obstacle cell above $d$, then we are in case~\ref{upd:case2} or~\ref{upd:case3} of \texttt{UpdateIS}, and $d^-$ is removed from $S_k$. 
    Thus invariant~\ref{inv:3c} also holds for all disks above~$d$.
This shows that the addition of a disk to $S_k$ in step~\ref{upd:ins-Sk} of \texttt{GreedyIS} maintains invariant~\ref{inv:3c}.

\item Suppose we add a disk $d$ associated with a newly created leaf $c$ of the nonatree in step~\ref{ins:obs-leaf} of \texttt{Insert} and step~\ref{del:obs-leaf} of \texttt{Delete}, then we call \texttt{UpdateIS} (in step~\ref{ins:call-update} of \texttt{Insert} and step~\ref{del:call-update} of \texttt{Delete}) for cell $c$. For the disks above~$c$, the argument above goes through and shows that 3-clearance is also maintained for all disks above~$d$.
\item Suppose that a disk $\beta(c^-)$ associated with a cell~$c_\beta$ was moved from $B_k$ to $S_k$ in step~\ref{upd:case3} of \texttt{UpdateIS}. We know that by invariant~\ref{inv:4d} the obstacle disk $o(c_\beta)$ is disjoint from all disks associated with cells~$c$ for which~$c_\beta\prec c$. Furthermore, in step~\ref{upd:case3} of \texttt{UpdateIS} the obstacle disk~$o(c)$ of the highest cell~$c\prec c_\beta$ associated with a disk in~$S_k$, does not intersect~$\beta(c^-)$. Hence, all obstacle disks associated with cells~$c$ for which $c\prec c_\beta$ are also disjoint from~$\beta(c^-)$ (by Lemma~\ref{lem:closest-obstacle}), and 3-clearance is maintained. 
\end{itemize}

    \paragraph{Invariant~\ref{inv:obs}.} To maintain invariant~\ref{inv:obs}, we take the following steps:
        In step~\ref{ins:ins-nonatree} of both \texttt{Insert} and \texttt{Delete} the structure of~$N_k$ can change, and hence we query~$F_o'$ and update $F_o$ and~$F_o'$ to make sure that invariant~\ref{inv:obsa} remains satisfied.
        Furthermore, when a disk is added to~$S_k$ in step~\ref{add:obs} of \texttt{Add}, or when a disk is deleted from~$S_k$ in step~\ref{remove:obs} of \texttt{Remove}, $F_o$ and~$F_o'$ are updated to ensure that invariant~\ref{inv:obsa} is satisfied. Whenever we add an obstacle cell to $F_o$ and $F_o'$ we also create an obstacle disk~$o(c)=3d'$, where~$d'$ is the smallest enclosing disk of~$c$, satisfying Invariant~\ref{inv:obsb}

    \paragraph{Invariant~\ref{inv:4}.}
    First consider invariant~\ref{inv:4c}. 
    Observe disks are added to~$B_k$ only in \texttt{UpdateIS}, and that \texttt{Insert}, and \texttt{Delete} only remove disks from~$B_k$. Every such modification of~$B_k$ (add or remove) is accompanied by a corresponding assignment change of~$\beta$ (resp. setting it to a disk, or to \textsc{nil}). The only other time the assignment changes, is in step~\ref{upd:case3}, and after changing the assignment of a barrier disk to a different cell, the assignment of the old cell is set to \textsc{nil}. Since \texttt{UpdateIS} is the only subroutine where $\beta$ is set to a non-\textsc{nil} value, invariant~\ref{inv:4c} is maintained when the following two conditions are met. Subroutine \texttt{UpdateIS} correctly assigns each barrier disk~$b\in B_k$ associated with cell~$c$ to the highest obstacle cell below $c$; no obstacle cell appears between a barrier disk and its assigned obstacle cell.
    
    In step~\ref{upd:case2} of \texttt{UpdateIS}, there is no barrier disk assigned to the bottom or top obstacle cell of the ascending path. Thus we can turn the top obstacle cell into a barrier cell, and assign the barrier disk to the highest obstacle cell below it, since there is not other assigned barrier disk in the two affected ascending paths. In step~\ref{upd:case3} there is only a barrier disk at the top of the ascending path, which will either be turned into a disk of~$S_k$, or is assigned to the highest obstacle cell below it. In both cases \texttt{UpdateIS} ensures that Invariant~\ref{inv:4c} is maintained.
    
    As explained in the previous paragraph, \texttt{UpdateIS} ensures that all assignments for barrier disks in the affected ascending paths are set correctly, regardless of the new obstacle cells that are created. It therefore suffices to show that no obstacle cells appear between barrier disks and their assigned obstacle cells in subroutines \texttt{Insert} and \texttt{Delete}. 
    
    During \texttt{Insert} a new (true) obstacle can appear in a newly created leaf, which has not obstacle cells below it, or new (merge) obstacle cells can be created in step~\ref{ins:ins-nonatree}. Such a merge obstacle cell splits an ascending path in two ascending paths, and we need to update the assignment when a barrier disk is assigned to the obstacle cell at the bottom of the lower ascending path, but the barrier disk is located in the top path. In step~\ref{ins:delete_Bk2} we set the assignment to \textsc{nil} and remove the barrier disk, when this happens, to maintain Invariant~\ref{inv:4c}.

    During \texttt{Delete} obstacle cells can only appear as new leafs in~$N_k$, and hence these cannot violate barrier disk assignment. Though it may happen that barrier disks, or their assigned obstacles disappear because disk deletion (resp. in step~ \ref{del:del-Bk} and steps~\ref{del:obs-unmerge} and \ref{del:del-Sk}). In all these cases the barrier disks are removed and assignments set to~\textsc{nil}. Thus, Invariant~\ref{inv:4c} is maintained
    
    For invariants~\ref{inv:bobs} and~\ref{inv:4d}, recall that new disks can be added to $B_k$ only in step~\ref{upd:case2} of \texttt{UpdateIS}, where a disk $d^-$ associated with a cell $c^-$ moves from $S_k$ to $B_k$. Before the dynamic update, $S_k$ satisfied invariant~\ref{inv:3b}, and so $o(c^-)$ is disjoint from all disks in $S_k$ associated with cells $c$, $c^-\prec c$. Thus, when~$d^-$ is moved to~$B_k$, it inherits~$o(c^-)$ as the barrier clearance disk with the same property, maintaining both invariants.

    \paragraph{Invariant~\ref{inv:5}.} Finally, we prove that invariant~\ref{inv:5} is also satisfied after a dynamic update. We first prove that after every insertion or deletion of a disk~$d$, the \texttt{UpdateIS} subroutine is called on each ascending path where invariant~\ref{inv:5} might possibly be violated, and then prove that \texttt{UpdateIS} restores the invariant.
    \begin{itemize}
        \item After an insertion of a disk $d$ into~$\mathcal{D}$, invariant~\ref{inv:5} can be violated in two ways: (1) disk~$d$ violates invariant~\ref{inv:5}; or the set of dominating disks decreases because of (2) the deletion of a barrier disk or an obstacle cell.
        When a disk $d$ is inserted, associated with a cell $c$, then \texttt{Insert} calls \texttt{UpdateIS} on the cell~$c$ and hence the ascending path containing~$c$ (in step~\ref{ins:call-update}). Subroutine \texttt{Insert} removes a barrier disk in step~\ref{ins:delete_Bk2}. Observe that in this case the barrier disk is located above a new merge obstacle cell~$c'$, and \texttt{UpdateIS} is called on the ascending path with~$c'$ as the bottom cell. 
        Finally, note that \texttt{Insert} does not remove any obstacle cells:  Steps~\ref{ins:obs-misc} and~\ref{ins:delete_Bk1} together may turn a true obstacle cell into a merge obstacle cell.
        \item After a deletion of a disk $d$ from~$\mathcal{D}$, invariant~\ref{inv:5} can be violated in only one way: if the set of dominating disks decreases because of the deletion of a barrier disk or an obstacle cell.
        Step~\ref{del:obs-leaf} of \texttt{Delete} removes a leaf $c$ of the nonatree, which is an obstacle cell, and an ancestor cell $c_p$ becomes a leaf. 
        Step~\ref{del:obs-unmerge} of \texttt{Delete} removes a merge obstacle cell, when~$c$ was the only cell in one of the two subtrees of ancestor cell~$c_p$.
        In both cases, every disk in $\mathcal{D}_k$ that was in the ascending path starting from $c$ is now in an ascending path containing $c_p$, on which we call \texttt{UpdateIS} in step~\ref{del:call-update}. 
        
        Step~\ref{del:del-Sk} of \texttt{Delete} removes a true obstacle cell (when $d\in S_k$), and any barrier assigned to this cell. Step~\ref{del:del-Bk} of \texttt{Delete} removes a barrier. In both cases, step~\ref{del:call-update} calls \texttt{UpdateIS} on the ascending path containing $c$ as well as any disks dominated by the obstacle disk or barrier clearance disk that has been removed. 
        
        \item We now show that \texttt{UpdateIS} re-establishes invariant~\ref{inv:5} on the ascending path $P$ it is called upon. Note that \texttt{UpdateIS} calls subroutine \texttt{GreedyIS} on a path $P'$, where either $P'=P$ (steps~\ref{upd:case1}--\ref{upd:case2}) or $P\subset P'$ (step~\ref{upd:case3}). Before the call to \texttt{GreedyIS} on $P'$, it removes any barrier disks (step~\ref{upd:empty-Bk}) or obstacle cells (step~\ref{upd:case3}) from interior cells of $P'$. The disks in $\mathcal{D}_k$ that are dominated by any obstacle disk or barrier clearance disk that are removed in these steps are all contained in path $P'$.  
        
        
        Then \texttt{GreedyIS} incrementally adds disks to~$S_k$ by querying $T_\cup$ to find disks that are not yet dominated by the highest obstacle disk in~$P'$. This process ends when a disk is found outside of~$P'$, and hence all disks in the ascending path are either in~$S_k$, or intersect the highest obstacle disk below them, satisfying Invariant~\ref{inv:5a} in both cases. After subroutine \texttt{GreedyIS} has finished, \texttt{UpdateIS} may still swap disks between $B_k$ and $S_k$ (in steps~\ref{upd:case2} and~\ref{upd:case3}), but this does not violate the invariant, since disks will simply swap from satisfying Invariant~\ref{inv:5a} to~\ref{inv:5b} or vice versa. Thus, Invariant~\ref{inv:5} is maintained.
        \qedhere
    \end{itemize}
\end{proof}

We have described how to maintain the independent sets $S_1,\ldots ,S_8$ satisfying invariants~\ref{inv:1}--\ref{inv:5} in polylogarithmic expected amortized update time. By Lemma~\ref{lem:invariants}, the largest of $S_1,\ldots , S_8$ is a constant-factor approximate MIS of $\mathcal{D}$. For each dynamic update in $\mathcal{D}$ can add or remove $O(1)$ disks in $S_k$, for some $k\in \{1,\ldots, 8\}$, by Lemmata~\ref{lem:insert-time}--\ref{lem:delete-time}. 
%
Thus, by Lemma~\ref{lem:mix}, we can smoothly transition from one independent set to another using the MIX algorithm, with $O(1)$ changes in the ultimate independent set per update in $\mathcal{D}$, and conclude the following theorem.

\theoremArbitraryDisks*

\subsection{Lower Bound}\label{sec:lb}
We state a lower bound for the \textsc{DGMIS} problem for a set of congruent disks. Our argument is similar to that of Henzinger et al.~\cite{Henzinger0W20}, who gave such lower bounds for hypercubes. For the sake of completeness, we state their result. 

\begin{theorem}\label{lb:disk}
For a fully dynamic set of unit disks in the plane, there is no algorithm for \textsc{DGIS} with approximation ratio $1+\varepsilon$ and amortized update time $n^{O((1/\varepsilon)^{1-\delta})}$, for any $\varepsilon,\delta>0$, unless the \textsf{ETH} fails.
\end{theorem}

\begin{proof}
Marx~\cite{marx2007optimality} showed that, assuming \textsf{ETH}, there is no $\delta>0$,
such that a $2^{(1/\varepsilon)^{O(1)}}\cdot n^{O((1/\varepsilon)^{1-\delta})}$ 
time \textsf{PTAS} exists for \textsc{MIS} for unit disks. Suppose we have an algorithm that maintains $(1+\varepsilon)$-approximate solution with an amortized update time $n^{O((1/\varepsilon)^{1-\delta})}$. Then, we could transform the input instance of \textsc{MIS} to a dynamic instance by inserting the disks one-by-one, in overall $n^{O((1/\varepsilon)^{1-\delta})}$ time. This contradicts the result of Marx~\cite{marx2007optimality}. 
\end{proof}

\section{Conclusions}
\label{sec:con}

We studied the dynamic geometric independent set problem for a collection of disks in the plane and presented the first fully dynamic algorithm with polylogarithmic update time. First, we showed that for a fully dynamic set of unit disks in the plane, a constant factor approximate maximum independent set can be maintained in polylogarithmic update time. Moreover, we showed that this result generalizes to fat objects in any fixed dimension. 
Our main result was a dynamic data algorithm that maintains a constant factor approximate maximum independent set in polylogarithmic amortized update time. One bottleneck in our framework is the nearest/farthest neighbor data structure~\cite{KaplanMRSS20, Liu22} (as discussed in Section~\ref{sec:cont}), which provides only \emph{expected amortized} polylogarithmic update time. This is the only reason why our algorithm does not guarantee deterministic worst-case update time, and it does not extend to balls in $\mathbb{R}^d$ for $d\geq 3$, or to arbitrary fat objects in the plane. It remains an open problem whether there is a dynamic nearest/farthest neighbor data structure in constant dimensions $d\geq 2$ with a worst-case polylogarithmic update and query time: Any such result would immediately carry over to a fully dynamic algorithm for an approximate MIS for balls in higher dimensions.  

\paragraph{Beyond Fatness.} While there have been several attempts to obtain constant-factor dynamic approximation schemes for various sub-families of rectangles, it is not known if, for a dynamic collection of axis-aligned rectangles in the plane, there exists an algorithm that maintains a constant-factor approximate maximum independent set in sublinear update time. On the one hand, due to Henzinger et al.~\cite{Henzinger0W20}, we know that it is not possible to maintain a $(1+\varepsilon)$-approximate solution in $n^{O((1/\varepsilon)^{1-\delta})}$ amortized update time, for any $\delta>0$, unless the \textsf{ETH} fails. On the other hand, recent progress on \textsc{MIS} for a static set of axis-parallel rectangles resulted in several constant-factor approximations \cite{galvez20223,mitchell2022approximating}. However, these algorithms are based on dynamic programming, and hence it is not clear how to naturally extend them into the dynamic realm.

\bibliographystyle{alphaurl}
\bibliography{main.bib}
\end{document}